\documentclass[letterpaper, 12pt, oneside]{book}

    \usepackage[left=108pt, right=74pt, top=108pt, bottom=81pt, dvips, pdftex]{geometry}
    
    \usepackage{fancyhdr} 
    \renewcommand{\bibname}{References}

    \usepackage{amsmath, amsthm, amsfonts,amssymb} 
    \usepackage{setspace} 
    \usepackage[linktocpage,bookmarksopen,bookmarksnumbered,
                pdftitle={UC Davis Ph.D. Doctoral Thesis},
                pdfauthor={Department of Mathematics},%
                pdfsubject={UC Davis Ph.D. Doctoral Thesis},%
                pdfkeywords={UC Davis Ph.D. Doctoral Thesis}]{hyperref}
    
     \usepackage{graphicx} 

    \newtheorem{theorem}{Theorem}[section]
    \newtheorem{proposition}[theorem]{Proposition}
    \newtheorem{lemma}[theorem]{Lemma}
    \newtheorem{corollary}[theorem]{Corollary}
    \theoremstyle{definition}
        \newtheorem{definition}[theorem]{Definition}

    \theoremstyle{remark}

    \numberwithin{equation}{section}
    \numberwithin{figure}{section}

    \newcommand{\n}{\vspace{12pt}} 
    \newcommand{\rem}[1]{{\bf Remark:}}
    \newcommand{\newchapter}[3] 
	{                           
        \chapter[#2]{#3}
        \chaptermark{#1}
        \thispagestyle{myheadings}
	}
\newcommand{\be}{\begin{equation}}
\newcommand{\ee}{\end{equation}}
\newcommand{\benn}{\begin{equation*}}
\newcommand{\eenn}{\end{equation*}}
\newcommand{\bea}{\begin{eqnarray}}
\newcommand{\eea}{\end{eqnarray}}
\newcommand{\beann}{\begin{eqnarray*}}
\newcommand{\eeann}{\end{eqnarray*}}
\newcommand{\bra}[1]{\langle #1|}
\newcommand{\ket}[1]{|#1\rangle}
\newcommand{\braket}[2]{\langle #1|#2\rangle}
\newcommand{\ketbra}[2]{|#1\rangle\langle #2|}
\newcommand{\pure}[1]{\ketbra{#1}{#1}}
\newcommand{\Tr}{\mathop{\mathrm{Tr}}}
\providecommand{\one}{\leavevmode\hbox{\small1\kern-3.8pt\normalsize1}}
\newcommand{\PI}{P_Y}
\newcommand{\PB}{P_B}
\newcommand{\PE}{P_{V\setminus Y}}

\newcommand{\PA}{P_A}
\newcommand{\PV}{P_{V\setminus A}}

\newcommand{\Ir}{\mathbb{Z}}
\newcommand{\Cx}{\mathbb{C}}
\newcommand{\A}{\mathcal{A}}

\newcommand{\diam}{{\rm diam}}
\def\idty{{\mathchoice {\rm 1\mskip-4mu l} {\rm 1\mskip-4mu l} %
{\rm 1\mskip-4.5mu l} {\rm 1\mskip-5mu l}}}

\newcommand{\B}{\mathcal{B}}

\newcommand{\X}{\mathcal{X}}

\newcommand{\rhoAB}{\rho_{\A\otimes\B}}
\newcommand{\rhoN}{\rho_{[1,n]}}
\newcommand{\rhoPN}{\rho_{1,[p,n]}}

\newcommand{\eof}[3]{E_{[#1,#2]}(#3)}
\newcommand{\E}{\mathbb{E}}
\newcommand{\Ehat}{\mathbb{\hat{E}}}
\newcommand{\FCS}{\textsf{FCS}}
\newcommand{\EoF}{\textsf{EoF}}
\providecommand{\bysame}{\leavevmode\hbox to3em{\hrulefill}\thinspace}
\newcommand{\etal}{\textit{et al.}}
\newcommand{\W}{\mathcal{W}_{\lambda,d}}
\newcommand{\D}{\mathcal{D}_{\lambda,d}}
\begin{document}
    

    \pagenumbering{roman}
    \pagestyle{plain}

    %
    %

    \singlespacing

    ~\vspace{-0.75in} 
    \begin{center}

        \begin{huge}
            Entanglement and Ground States of Gapped Hamiltonians
        \end{huge}\\\n
        By\\\n
        {\sc Spyridon Michalakis}\\
        B.S. (Massachusetts Institute of Technology) 2003\\
        \n
        DISSERTATION\\\n
        Submitted in partial satisfaction of the requirements for the degree of\\\n
        DOCTOR OF PHILOSOPHY\\\n
        in\\\n
        MATHEMATICS\\\n
        in the\\\n
        OFFICE OF GRADUATE STUDIES\\\n
        of the\\\n
        UNIVERSITY OF CALIFORNIA\\\n
        DAVIS\\\n\n
        
        Approved:\\\n\n
        
        \rule{4in}{1pt}\\
        ~Bruno Nachtergaele(Chair)\\\n\n
        
        \rule{4in}{1pt}\\
        ~Craig Tracy\\\n\n
        
        \rule{4in}{1pt}\\
        ~Greg Kuperberg\\
        
        \vfill
        
        Committee in Charge\\
        ~2008

    \end{center}

    \newpage

    %
    %
    
   \begin{center}
    To mom and dad, for your unconditional love and support. \\
    To Nikos and Marios, for being my other $\frac{2}{3}$ ...
    \end{center}
    
     ~\\[6.75in] 
    \centerline{
                \copyright\ Spyridon Michalakis,
                            2008. All rights reserved.
               }

    \newpage

    \doublespacing

    %
    %
    
    \tableofcontents
    
    \newpage

    %
    %
    
    %
     \thispagestyle{empty}

    ~\vspace{-1in} 
    \begin{flushright}
        \singlespacing
        Spyridon Michalakis\\
        ~September 2008\\
        Mathematics
    \end{flushright}

    \centerline{\large Entanglement and Ground States of Gapped Hamiltonians}
    
    \centerline{\textbf{\underline{Abstract}}}
    
This thesis weaves together three separate results whose common thread is the study
of quantum entanglement, a physical resource which, much like energy, may be exploited to perform tasks with the potential for extraordinary applications in areas of security, computing and simulation of classically intractable quantum phenomena.

We begin by considering entanglement properties of an important class of quantum states, introduced by Fannes, Nachtergaele and Werner, known as Finitely Correlated States (\FCS).
We derive bounds for the entanglement of a spin with an (adjacent and
non-adjacent) interval of spins in an arbitrary pure \FCS. Finitely Correlated
States are otherwise known as matrix product states or generalized valence-bond
states. The bounds we derive become exact in the case where one considers the
entanglement of a single spin with a half-infinite chain to the right (or the 
left) of it. Our bounds provide a proof of the recent conjecture by Benatti,
Hiesmayr, and Narnhofer that their necessary condition for non-vanishing
entanglement in terms of a single spin and the ``memory'' of the \FCS, is also
sufficient \cite{BHN}. Our result also generalizes the study of entanglement in
the ground state of the AKLT model by Fan, Korepin, and 
Roychowdhury~\cite{FKR}. Furthermore, our result permits a more efficient calculation,
numerically and in some cases even analytically, of the  entanglement in
arbitrary finitely correlated quantum spin chains.

We continue the study of entanglement in the setting of ground states of Hamiltonians with
a spectral gap. In particular, for $V$ a finite subset of $\Ir^d$, we let $H_V$ denote a 
Hamiltonian on $V$ with finite range,
finite strength interactions and a unique ground state with a non-vanishing spectral gap.
For a density matrix $\rho_A$ that describes the finite-volume restriction to a region $A$
of the unique ground state, we provide a detailed version of Hastings' proof in one dimension \cite{hastings_law}, 
that the entropy of $\rho_A$ is bounded by a uniform constant $C$, where $C$ depends only on the interaction strength, the spectral gap and the maximum among the dimensions of the state spaces associated 
with each site in $V$. Moreover, we provide a detailed generalization of the  $1$-dimensional construction of Hastings' approximation to the ground state in dimensions $2$ and higher, which may prove useful for understanding the underlying structure of ground states of gapped Hamiltonians in higher dimensions.

Finally, we turn our attention to the study of a conjecture central to Quantum Information Theory, the multiplicativity of the maximal output Schatten $p$-norm of quantum channels, for $p > 1$. In particular, we study the output $2$-norm for a special class of quantum channels, the depolarized Werner-Holevo channels, and show that multiplicativity holds for a product of two identical channels in this class. Moreover, it is shown that the depolarized Werner-Holevo channels do not satisfy the entrywise positivity (EP) condition introduced by C. King and M.B. Ruskai, which suggests that the main result is non-trivial, since the EP condition has been shown to imply multiplicativity of the output $2$-norm.

    \newpage

    %
    %

    \chapter*{\vspace{-1.5in}Acknowledgments and Thanks}

       I begin this section by saying {\it thank you} to Annemarie Sheets. Thank you for spending four difficult and wonderful years at MIT by my side, for motivating me to go to graduate school to get my Ph.D. and finally, for putting up with my mathematical temperament. 

Talking about temperament, I want to thank my adviser, Bruno Nachtergaele, for having the kind of patience with me that comes with true wisdom and a deep devotion to helping others fulfill their potential. I will always think of him as my mentor, teacher and friend. At this point, I would also like to thank the faculty and staff of a department of mathematics that deserves to be ranked highly for its quality of research and sense of community. Many here don't just
do math, they live their lives as if mathematics truly matters to them. In that spirit, I wish to thank everyone who supported the idea of giving back to the community through the Explore Math program.
Brandy, Yvonne and Sarah, the kids will never forget your passion for teaching them about the beauty of mathematics.

As a graduate student and specifically a researcher, I wondered if there would be one day that I would wake up and feel like I was finally a mathematician. That day has not come yet for me and I know because I have met {\it mathematicians}.
Whether it be my adviser, or the fiercely intelligent and entertaining Prof. Greg Kuperberg who introduced me to quantum information theory, I had the honor of learning from some of the best in the field. Through
generous funding from the department of mathematics, I flew to conferences and workshops around the world to meet and learn from some of the top researchers in mathematical physics. I think it is only fitting that I take a few lines to thank F. Benatti, H. Narnhofer and P. Horodecki for valuable discussions, during a visit to Trieste, Italy when the article that is now Chapter \ref{sec:LabelForChapter2} was being prepared for publication. Mary Beth Ruskai for suggesting the question that is now Chapter \ref{sec:LabelForChapter4} and the many helpful discussions I had with her over the years. 
I would also like to
thank E. Hamza, R. Sims and B. Nachtergaele, for their invaluable contributions to Chapter \ref{sec:LabelForChapter3} and to M. Hastings for his ground breaking research and for accepting me under his wings as I learn to fly as an independent researcher in the next few years. Moreover, it is a great honor to have Prof. Craig Tracy on my thesis committee and for that I am very thankful.
Of course, the support of several NSF grants made it possible for an international student like myself to do research in such an exciting field, so with great gratitude I acknowledge support from NSF Grants \#DMS-0303316 and \#DMS-0605342.

My social life here at Davis would not have been such a wonderful experience without the presence of an amazing circle of friends. So, Zach, thank you for teaching me Japanese and playing soccer with me when math was straining my mind. Tu, thank you for inviting me over at your table one fateful night, to eat with you and Ann and Nakul. We have been inseparable since then and I feel truly lucky to have met you all. To Vivian, thank you for running with me, and cooking for me and staying up with me late into the night while I was doing math inside a sleepy looking math building. 

To my wonderful family, this one is for you.

\begin{flushright}
Davis, CA\\
July, 2008
\end{flushright}
    
    \newpage

    %
    %

    \pagestyle{fancy}
    \pagenumbering{arabic}
       
    %
    %

    \newchapter{Introduction}{Introduction}{Introduction}
    \label{sec:Intro}
	\section[Informal Overview and Historical Motivation]{Informal Overview and Historical Motivation}
        \label{sec:Intro:Motivation}

The newly developed fields of Quantum Information Theory (QIT) and Quantum Computation (QC) have cross-fertilized a variety of areas in
long established fields of mathematics and physics, such as Operator Theory and Condensed Matter Physics. The mysterious resource of {\it quantum entanglement}, the stronger-than-classical non-local correlation between quantum particles, is at the center of many exciting new theories and applications, such as polynomial-time factoring and quantum teleportation. 
The physical and mathematical framework behind many of these applications can be formulated in the language of one dimensional arrays of quantum particles (spins), known as {\it quantum spin chains}. These systems are completely described by their associated Hamiltonian, a mathematical object which encodes the interactions between the particles. More importantly, the Hamiltonian contains information about the possible states of the system it describes and the energy level associated with each such state. Of special importance is the state of the system with the lowest energy, often called the {\it ground-state} of the system. The minimum amount of energy required to excite the system from its ground-state is called the {\it spectral gap}.

Not long ago, quantum spin chains were primarily the object of study in Condensed Matter Physics. Since the advent of QIT and QC, there has been a renewed interest in the properties of these seemingly simple quantum objects. The dynamics that govern spin chains can be applied to quantum states to transform them to desired target states, thus achieving quantum computation. The existence of a spectral gap in the Hamiltonian describing the dynamics of the spin chain is key to performing quantum computation reliably. The gap not only acts as a safeguard against external perturbations that might otherwise derail the computation, but it also determines the time it takes computational protocols such as adiabatic quantum evolution \cite{adiabatic} to solve problems which are intractable for classical computers. Hence, the study of the spectral gap and its implications is central to efficient and reliable quantum computation.

Recent progress in the study of the velocity with which interactions spread between spins on a spin chain \cite{ns,loc} has been combined with techniques used for splitting Hamiltonians with a spectral gap into ``frustration free'' components \cite{gap_vbs,solvgap}, to produce some spectacular results in Quantum Information Theory, with applications to Quantum Computation. A question central to the study of the complexity of simulating quantum systems with classical resources \cite{complexity,aharonov} is the conjecture that the entropy of a bounded region of spins in the ground state of a gapped Hamiltonian grows proportionally to the surface area of the region and not the bulk (volume). This conjectured entropy scaling has attracted much attention lately as it may both help explain and optimize the running time of algorithms \cite{simutability}, such as the Density Matrix Renormalization Group (DMRG) algorithm \cite{DMRG}, that exploit the underlying structure of ``frustration free" ground states, such as the one describing the ground state of the well-known AKLT model \cite{haldane, AKLT}, to calculate efficiently certain properties of a given Hamiltonian. The reason why such algorithms work so well at calculating properties of quantum systems is straightforward: The system, usually, has a limited amount of entanglement between its components, hence it can be approximated efficiently by a series of ever-refined classical-like components, known as {\it Finitely Correlated States} \cite{FNW,FNW2}, or more recently, as Matrix Product States \cite{VC,MPS}.
This special class of states will continue to attract much attention, as we strive to understand better
the conditions under which entanglement persists in the ground states of quantum systems we wish
to use in practice.

\section{Summary of the Main Results}
\label{sec:Intro:Summary}

In the chapters that follow, we present research that has either been published (Chapter \ref{sec:LabelForChapter2} in \cite{spiros_prl} and Chapter \ref{sec:LabelForChapter4} in \cite{spiros_jmp}), or, is being prepared for publication (Chapter \ref{sec:LabelForChapter3}).
 
In particular, in Chapter \ref{sec:LabelForChapter2} we present a proof of a conjecture posed by Benatti et al. in \cite{BHN}. In their paper, using the structure of  Finitely Correlated Pure States, introduced by Fannes et al. in \cite{FNW2,FNW}, the authors construct a ``dual'' state $\mathbb{F}(\rho)$ that encodes a large part of the information found in the original state of the spin-chain.
Using a brilliantly simple argument, they proceed to show that in order for entanglement to exist between a spin at position $1$ of the spin-chain and a subset of spins in positions $[p,n]\, (p>1)$ in the original state, it is a necessary condition that the "dual" state $\mathbb{F}(\rho)$ has non-zero entanglement. Using numerical evidence, they conjecture then that this condition is also sufficient for entanglement between subsets of spins on the translation invariant state. In this chapter, we resolve this conjecture by proving the stronger claim that the entanglement of formation between the first spin and the spins at sites $[2,n]$ converges exponentially fast in $n$ to the entanglement of formation of $\mathbb{F}(\rho)$.
\vskip 1em

In the first section of Chapter \ref{sec:LabelForChapter3}, we turn our focus to the Area Law for $1$-dimensional ground states of gapped Hamiltonians and provide a detailed proof of a ground breaking result by Hastings \cite{hastings_law}, the proof of the conjectured area law for $1$-dimensional Hamiltonians with a spectral gap. In section \ref{sec:LabelForChapter3:Section2}, we leave the realm of Quantum Information Theory and enter the field of Quantum Statistical Mechanics in order to construct a higher-dimensional generalization of Hastings' approximation for the ground state of a gapped Hamiltonian, which is central to the proof of the area law in one dimension. It is our hope that a clear presentation of the techniques and ideas involved in the construction of such an approximation, will shed more light on the underlying entanglement of the ground states of Hamiltonians that describe quantum systems with a spectral gap \cite{gap_vbs,shannon,knabe}, leading to a better understanding of the scaling of entropy in $2$ and higher dimensions.

The results in Chapter \ref{sec:LabelForChapter4} came about after a conversation with M. B. Ruskai at the 2006 International Congress of Mathematical Physics in Rio de Janeiro, where she posed a question related to the ``multiplicativity conjecture'' for a certain class of quantum channels. In mathematical terms, a quantum channel is a completely positive, trace-preserving map between the
algebras of bounded operators on Hilbert spaces, and ``multiplicativity'' refers to the maximal output (Schatten) $p$-norm equality $\|\Phi\otimes\Psi\|_p = \|\Phi\|_p\,\|\Psi\|_p$, where $\Phi, \Psi$ denote quantum channels and $\|\Omega\|_p = \sup_{\rho} \|\Omega(\rho)\|_p$, $\rho$ being the input state. 
The study of this conjecture and its implications are at the center of Quantum Information Theory, as other important ``additivity'' conjectures involving the capacity of quantum channels and the role of entanglement in transmitting quantum information may be proven \cite{shor,fukuda} by showing the following sufficient condition: For each quantum channel $\Phi$, there exists a $p(\Phi) > 1$ such that $\|\Phi\otimes\Phi\|_{p_n} = \|\Phi\|_{p_n}^2$ for $1 < p_n < p(\Phi)$ and $p_n \rightarrow 1$.
Focusing on the special case of the maximal output $2$-norm (a.k.a. Hilbert-Schmidt norm), I was able to prove multiplicativity for a class of channels given by $\Phi(\rho) = \lambda \rho + (1-\lambda) \frac{\one_d - \rho^T}{d-1}$, for all $d \ge 2, \lambda \in [0,1]$ \cite{spiros_jmp}. The channels I studied are a generalization of the well-known Werner-Holevo channels $W(\rho) = \frac{\one_d - \rho^T}{d-1}$, which gave the first counterexamples to multiplicativity (for $p > 4.79$).
It was recently shown that for $p \ne 2$ there exist quantum channels that yield counterexamples to the conjecture of multiplicativity \cite{Winter, Hayden}. To the best of my knowledge, there is no published counterexample for $p=2$, but the focus has shifted to studying multiplicativity in the region of  $p < 1$. 
In any case, like with so many other questions in Quantum Information Theory, the study of multiplicativity and entanglement in general, promises to be both fruitful for various areas of mathematics, physics and computer science and truly exciting in the development of both the theory and future applications!

        
        %
    
        

    %
    %

    \newchapter{Entanglement in Finitely Correlated States}{Entanglement in Finitely Correlated Spin States}{Entanglement in Finitely Correlated Spin States}
    \label{sec:LabelForChapter2}

        
    
       
\section{Introduction}
Entanglement properties of quantum spin-chains have recently attracted
attention from researchers in quantum information theory and condensed matter
physics. From the perspective of quantum information theory, the distribution
of entanglement over long ranges via local operations on a
spin-chain~\cite{bravyi2006, EO, NOS} has obvious applications to teleportation-based
models of quantum computation ~\cite{Loc-Ent, CR, GB}. Moreover, it has
recently been shown that entanglement in \textit{finitely correlated
chains}~\cite{FNW2} can be used to achieve universal quantum
computation~\cite{VC} and provide a computational tool for adiabatic quantum
computation~\cite{MPS}. On the other hand, the scaling behavior of quantum
correlations in infinite spin-chains is intimately related to their critical
behavior (recent work has established a general mathematical framework for
studying entanglement in infinite quantum spin-chains~\cite{KMSW}.)

Finitely Correlated States (\FCS) are a generalization of the so-called Valence
Bond Solid (VBS) states, which arise as the exact ground states of a
considerable variety of quantum spin Hamiltonians \cite{gap_vbs,AKLT}.
Interestingly, \FCS\ also provide approximate ground states of {\em any}
quasi-onedimensional spin system with finite-range interactions \cite{ABU}. In
particular, Density Matrix Renormalization Group calculations produce numerical
approximations of the ground state of spin chains that can be regarded as \FCS\
\cite{OS}. 

Motivated by the potential applications of distributed entanglement in finitely correlated chains, Benatti, \etal{} in~\cite{BHN}, 
give a {\bf necessary} condition for entanglement between a spin and a subset of other spins; namely, that the entanglement between a spin and, what one may think of as \cite{kuperberg}, the ``memory'' of the finitely correlated state must be non-zero.
They, furthermore, conjecture that the same condition is {\bf sufficient}, in the sense that
it implies entanglement between a spin and a subset of other spins.
We present here a proof of that conjecture by showing that the entanglement between a spin and its neighbors converges exponentially fast (in the number of neighboring spins) to the entanglement 
between a spin and the ``memory'' of the finitely correlated state. Moreover, we show that entanglement
between distant spins vanishes exponentially fast in the length of their separation.

Since finitely correlated states provide the exact ground states for
generalized valence-bond solid models~\cite{gap_vbs}, our result
generalizes the calculation of entanglement~\cite{FKR} for the AKLT
model~\cite{AKLT}.

More importantly, our result implies a simple and computationally efficient way for detecting distributed entanglement in finitely correlated states. Namely, the Positive Partial Transpose (PPT or Peres-Horodecki) criterion~\cite{Peres, H3} can be applied to the state describing the interactions of a spin with the "memory" of the finitely correlated state, to detect entanglement between a spin and a subset of other spins.

\section{The setup and main result} 
We will work with translation invariant pure \FCS~\cite{FNW2} on the infinite 
one-dimensional lattice. 
For each $i\in \mathbb{Z}$, the spin at site $i$ of the chain will be described by
the algebra $\A$ of $d\times d$ complex matrices. The observables of the spins
in an interval, $[m,n]$, are given by the tensor product $\A_{[m,n]} 
=\otimes_{j=m}^{n} (\A)_{j}$.
The algebra $\A_\mathbb{Z}$ describing the infinite chain arises as a suitable
limit of the {\it local} tensor-product algebras ${\A}_{[-n,n]}:=\otimes_{j=-n}^n(\A)_j$.
Any translation invariant state $\omega$ over $\A_\mathbb{Z}$ is completely
determined by a set of density matrices $\rho_{[1,n]}$, $n\geq 1$, which describe
the state of $n$ consecutive spins. In the case of a pure \FCS, as was
shown in \cite{FNW}, these density matrices can be constructed as follows:

The {\it memory}, $\B$, of a \FCS{} is represented by the algebra of 
$b\times b$ complex matrices. Let $\mathbb{E}:\A\otimes\B\mapsto\B$ 
be a completely positive unital map of the form
$\mathbb{E}(A\otimes B)=V(A\otimes B)V^{\dagger}$, where 
$V:\mathbb{C}^{d}\otimes\mathbb{C}^{b} \mapsto \mathbb{C}^{b}$, is a linear
map such that $VV^\dagger=\one_\B$. 
We define the completely positive map $\Ehat : \B \mapsto \B$, by $\Ehat(B) 
= \mathbb{E}(\one_{\A} \otimes B)$. The condition on $V$ implies that
$\Ehat$ is unital: $\Ehat(\one_\B)=\one_\B$.
In \cite{FNW} it is proved that for any pure translation invariant \FCS, 
it is always possible to choose
$\mathcal{B}$ and $V$ such that there is a unique, non-singular 
$b\times b$ density matrix, $\rho$, with the property 
$\Tr\rho\, \Ehat(B)=\Tr \rho\, B$, for all $B\in \B$.

We introduce the density matrix $\rhoAB$ associated with the state 
encoding the interaction between the spin at site $1$ and the ``memory''
of the \FCS:
\be
{\Tr}_{\A \otimes \B}\Bigl(\rhoAB\, A\otimes B\Bigr) 
= {\Tr}_{\B}\Bigl(\rho\, \mathbb{E}(A\otimes B)\Bigr).
\nonumber
\ee
Using the cyclicity of trace we also have $\rhoAB = V^{\dagger} \rho V$.

We are now ready to define the density matrices $\rhoN$ recursively, by the following identity:
\begin{equation}\label{def:rhoS}
\Tr (\rhoN A_1\otimes A_2 \otimes \cdots \otimes A_n) = {\Tr}_{\B} (\rho \, \E(A_1 \otimes \E(A_2 \otimes \cdots \otimes \E(A_{n-1} \otimes \E(A_n \otimes \one_{B}))\cdots))).
\end{equation}
From the above definition and the cyclicity of the trace we get the equivalent definition
\benn
\rhoN = {\Tr}_{\B}(V_{n}^{\dagger}\rhoAB V_{n}),
\eenn
where 
$V_{n} = (\one_{\A}\otimes V)(\one_{\A^{\otimes 2}}\otimes V) \cdots (\one_{\A^{\otimes n-1}}\otimes V)$.

An important property, intimately related to the exponential
decay of correlations in a {\it pure} \FCS{}
is that the peripheral spectrum of $\Ehat$ is trivial; 
that is, $\one_{\B}$ is the 
only eigenvector of $\Ehat$ with eigenvalue of modulus $1$ \cite{FNW}. 
This implies that the iterates of $\Ehat$ converge exponentially fast
to $\Ehat^\infty$ given by $\Ehat^{\infty}(B) = \lim_{n \rightarrow \infty} \Ehat^n(B) = \Tr(\rho B)\one_{\B}$.
More precisely, for any $\lambda$ such that $|\lambda_{i}| < \lambda < 1$,
for all eigenvalues $\lambda_{i}$ of $\Ehat$ different from $1$, there
exists a constant $c$ such that for all $n\geq 1$:
\be\label{Ehat:bound}
\|\Ehat^{n} - \Ehat^{\infty}\| \leq c \lambda^{n},
\ee
where the norm is the $\infty$-norm on $\B$ considered as a Banach space with the $1$-norm.

Our object of study is the entanglement of formation, \EoF~\cite{BDSW}.
The \EoF{} is defined for states of composite systems with 
a tensor product algebra of observables $\X_{1}\otimes\X_{2}$.
\begin{definition}[Entanglement of Formation]
The entanglement of formation of a bipartite state over $\X_{1}\otimes\X_{2}$ with associated density 
matrix $\sigma_{12}$ is given by:
\benn
\eof{\X_{1}}{\X_{2}}{\sigma_{12}} = \inf \sum_{i} p_{i} \, 
S\Bigl({\Tr}_{\X_{2}}(\sigma_{12}^{i})\Bigl),
\eenn
where $S(\rho) = - \Tr \rho \log \rho$ is the von Neumann entropy and the infimum of the average entropy is taken over all convex decompositions $\sigma_{12} = \sum_{i} p_{i}\, \sigma_{12}^{i}$ into pure states.
\end{definition}
Whenever $\X_{1}$ is finite dimensional, as will be the case for us,
the infimum can be replaced by a minimum in the above definition, i.e.,
there is an optimal decomposition, $\{p_{i},\,\sigma_{12}^{i}\}$, where the 
infimum is attained (see~\cite{EoF} for details).
We call $\{\ket{\phi_{i}}\}$ an {\bf ensemble} for the density matrix $\sigma$ whenever the latter can be
decomposed as $\sigma = \sum_{i} \pure{\phi_{i}}$.
There are an infinite number of ensembles corresponding to a given density matrix.
The following lemma provides us with a complete classification:
\begin{lemma}[Isometric Freedom in Ensembles,~\cite{Schroedinger,HJW}]\label{lem:iso-freedom}
Let $\{\ket{e_{i}}\}_{i=1}^{d}$ be the ensemble corresponding to the eigen-decomposition of
the density matrix $\sigma$, where $d = \text{rank}(\sigma)$. Then, $\{\ket{\psi_{i}}\}_{i=1}^{m}$ is
an ensemble for $\sigma$ if and only if there exists an isometry 
$U: \mathbb{C}^{d}\mapsto \mathbb{C}^{m}$
such that $$\ket{\psi_{i}} = \sum_{j=1}^{d} U_{i,j}\,\ket{e_{j}}, \, 1\leq i \leq m.$$
\end{lemma}
\noindent
The above lemma implies that any two ensembles for the same density matrix,
$\{\ket{\psi_{j}}\}_{j=1}^{M_{1}}$, and $\{\ket{\phi_{i}}\}_{i=1}^{M_{2}}$,
are similarly related via a partial isometry  
$W: \mathbb{C}^{M_{1}}\mapsto\mathbb{C}^{M_{2}}$.

Our main result is the following theorem:
\begin{theorem}\label{thm:conjecture}
For any pure translation invariant \FCS{} we have
\begin{align}
0 \leq \eof{\A}{\B}{\rhoAB} - \eof{\A}{\A^{\otimes n-1}}{\rhoN} \leq \epsilon(n),\label{eq:conjecture}
\end{align}
where $\epsilon(n)$ decays exponentially fast in $n$.
\end{theorem}

\section{Proof of the Theorem}
The lower bound is proven in~\cite{BHN}. For the sake of completeness, we include here the following proof. 

The definition of $\rhoN$ implies that every decomposition of $\rhoAB$ into pure states induces
a decomposition of $\rhoN$. Moreover, the restrictions to the spin at site $1$ of the $i$-th state
in the corresponding decompositions of $\rhoAB$ and $\rhoN$ are equal. To see this, note
that since the operators $V_n$ leave the first spin invariant, the cyclicity of the trace implies
$${\Tr}_{\A^{\otimes n-1}}(\rhoN^i) = {\Tr}_{\A^{\otimes n-1} \otimes \B}(V_{n}^{\dagger} \rhoAB^i V_n) = {\Tr}_{\B}(\rhoAB^i),$$
where we have used $V_n V_{n}^{\dagger} = \one_\A \otimes \one_\B$.
It follows that for each decomposition of $\rhoAB$ there is a corresponding decomposition of $\rhoN$ with equal average entropy. Since the average entropy of $\rhoN$ is minimized over a (possibly) larger set of decompositions, the lower bound follows.

We now focus on the upper bound.
We start with the following decompositions of $\rhoAB$ and $\rhoN$ into (unnormalized) pure states:
\begin{eqnarray}
\rhoAB &=& \sum_{i=1}^{b} V^{\dagger}\ket{\chi_{i}} \bra{\chi_{i}}V\label{rhoAB:dec}\\
\rhoN &=& \sum_{i,j=1}^{b} G_{n,j}^{\dagger}V^{\dagger}\ket{\chi_{i}} \bra{\chi_{i}}V G_{n,j}, \label{rhoN:dec}
\end{eqnarray}
where $\{\ket{\chi_i}\}_{i=1}^b$ is the eigen-ensemble of $\rho$ and
$G_{n,j} = V_{n} (\one_{\A^{\otimes n}}\otimes \ket{\chi_j}/ \|\chi_j\|)$.
The term in parenthesis in the expression for $G_{n,j}$ comes from the Kraus operators in the
decomposition of the completely positive map ${\Tr_{\B}}$.

By the observation following Lemma \ref{lem:iso-freedom}, we have that 
the (unnormalized) states $\ket{\Phi^{n}_{l}}$ in the {\bf optimal decomposition}  of $\rhoN$ are given by:
\be
\ket{\Phi^{n}_{l}} = \sum_{i,j=1}^{b} U_{l,(i j)} G_{n,j}^{\dagger}V^{\dagger}\ket{\chi_{i}},\quad 1\leq l \leq L\label{optimal},
\ee
for some partial isometry $U: \mathbb{C}^{b^{2}}\mapsto \mathbb{C}^{L}$,
whose dependence on $n$ we suppress.
Moreover, it is easy to check that $\rhoAB$ has a decomposition into $\sum_l \sigma_l$,  with $\sigma_{l} = \sum_j \pure{\Psi_l(j)}$ and
\be
\ket{\Psi_l(j)} = \sum_{i=1}^b u_{l,(ij)}\, \sqrt{p_i} \, V^{\dagger}\ket{\chi_i},\quad 1 \leq l \leq L
\label{rhoAB:optimal}
\ee

To calculate the \EoF{} we need the restrictions of 
$\{\pure{\Phi^{n}_{l}}\}$ and
$\{\sigma_{l}\}$ to $\A$:
\be
\tilde{\phi}^{n}_{l} = {\Tr}_{\A^{\otimes n-1}}(\pure{\Phi^{n}_{l}}) ,\quad
\tilde{\sigma}_{l} = {\Tr}_{\B}(\sigma_l).\label{rhoAB:rstr}
\ee
Define the density matrices 
$\phi^{n}_{l}=\tilde{\phi}^{n}_{l}/\alpha_{l}^{n}$
and $\sigma_{l}=\tilde{\sigma}_{l}/\beta_{l}$, where
$
\alpha_{l}^{n} \equiv \|\tilde{\phi}_{l}^{n}\|_{1} 
= \Tr(\tilde{\phi}_{l}^{n}),\quad
\beta_{l} \equiv \|\tilde{\sigma}_{l}\|_{1} \,\,= \Tr(\tilde{\sigma}_{l}).
$

From the definition of the \EoF{} and the optimality of $\{\tilde{\phi}^{n}_{l}\}_{l=1}^{L}$ we get:
\be
\eof{\A}{\B}{\rhoAB} - \eof{\A}{\A^{\otimes n-1}}{\rhoN}
\leq \sum_{l=1}^{L}\epsilon_{l}(n)\label{eq:main1},
\ee
where $\epsilon_{l}(n) = \beta_{l} S(\sigma_{l}) - \alpha_{l}^{n} S(\phi^{n}_{l}).$

It remains to show that $\sum_{l=1}^{L}\epsilon_{l}(n)$ is exponentially small.
We estimate each term in the sum as:
\be
|\epsilon_{l}(n)| \leq \beta_{l} |S(\sigma_{l}) - S(\phi^{n}_{l})| 
+ |\beta_{l}-\alpha_{l}^{n}| \log d\label{epsilon:first-bound},
\ee
since $\text{rank}(\phi^{n}_{l}) \leq d$.

To bound $|S(\sigma_{l}) - S(\phi^{n}_{l})|$ we use Fannes' inequality for the
continuity of the von Neumann entropy~\cite{Fannes}:
\be
|S(\sigma_{l}) - S(\phi^{n}_{l})| \leq (\log d+2)\|\sigma_{l}-\phi^{n}_{l}\|_{1} + \eta(\|\sigma_{l}-\phi^{n}_{l}\|_{1}),
\ee
where $\eta(x) = -x\log x$ and $\log$ is the natural logarithm.
By the triangle inequality we have:
\be
|\beta_{l}-\alpha_{l}^{n}| =
|\|\tilde{\sigma}_{l}\|_{1}-\|\tilde{\phi}_{l}^{n}\|_{1}| \leq \|\tilde{\sigma}_{l}-\tilde{\phi}^{n}_{l}\|_{1}\label{norms:bound}.
\ee
Another application of the triangle inequality gives:
\be
\|\sigma_{l}-\phi^{n}_{l}\|_{1} 
\leq \frac{\|\beta_{l}\sigma_{l}-\alpha_{l}^{n}\phi^{n}_{l}\|_{1}
+\|(\alpha_{l}^{n}-\beta_{l})\phi_{l}^{n}\|_{1}}{\beta_{l}}\nonumber,
\ee
which simplifies, with the use of~(\ref{norms:bound}), to the following inequality:
\be
\|\sigma_{l}-\phi^{n}_{l}\|_{1} \leq 2\frac{\|\tilde{\sigma}_{l}-\tilde{\phi}^{n}_{l}\|_{1}}{\beta_{l}}
\label{density:bound}
\ee
Combining equations~(\ref{epsilon:first-bound})-(\ref{density:bound}) and 
setting 
\be
\tau_{l}^{n} \equiv \|\tilde{\sigma}_{l}-\tilde{\phi}^{n}_{l}\|_{1}/\beta_{l}, \label{tau:first-bound}
\ee
we get the following bound for $\epsilon_{l}(n)$:
\be
|\epsilon_{l}(n)| \leq \beta_{l}[(\log d^{3}+4)\tau_{l}^{n} + \eta(2\tau_{l}^{n})]
\label{epsilon:bound}.
\ee
where we have assumed that $2\tau_{l}^{n} \leq 1/e$, 
to assure $\eta(x)$ is increasing.

To complete the proof, we show that $\tau_{l}^{n}$ is exponentially
small for large $n$. Since each $G_{n,j}$ leaves the spin at site $1$ invariant, the cyclicity of the trace yields:
\begin{equation}
\tilde{\phi}^{n}_{l} = \sum_{i,i',j,j'=1}^{b}
U_{l,(i' j')}^{*}U_{l,(i j)}{\Tr}_{\B}(V^{\dagger}\ket{\chi_{i}}\bra{\chi_{i'}}V G_{n,j'}G_{n,j}^{\dagger}),
\nonumber
\end{equation}
But $G_{n,j'}G_{n,j}^{\dagger} = \one_{\A} \otimes \Ehat^{n-1}(\ket{\chi_{j'}}\bra{\chi_{j}})/(\|\chi_{j'}\| \|\chi_{j}\|)$. Substituting
$\Ehat^\infty$ for $\Ehat^{n-1}$ we get:
$$\tilde{\sigma}_{l}-\tilde{\phi}^{n}_{l} = \sum_{i,i',j,j'=1}^b
U_{l,(i' j')}^{*} U_{l,(i j)} {\Tr}_{\B}(X_{i,i'}Y_{j,j'}),$$
where $X_{i,i'} = V^{\dagger}\ket{\chi_{i}}\bra{\chi_{i'}}V$ and
$Y_{j,j'} = \one_{\A} \otimes [\Ehat^{\infty}-\Ehat^{n-1}](\ket{\chi_{j'}}\bra{\chi_{j}})/(\|\chi_{j'}\| \|\chi_{j}\|)$.

Like all trace preserving quantum operations, the partial trace is contractive with respect to the $1$-norm. Hence, an application of the triangle inequality for the $1$-norm gives:
\begin{eqnarray*}
\|\tilde{\sigma}_{l}-\tilde{\phi}^{n}_{l}\|_{1} \leq \sum_{i,i',j,j'}^b
|U_{l,(i' j')}^{*}| |U_{l,(i j)}| \|X_{i,i'}\|_1 \|Y_{j,j'}\|_1
\end{eqnarray*}
It is not hard to see that
\be
\|X_{i,i'}\|_1 = \|\chi_{i}\| \|\chi_{i'}\|, \quad  \|Y_{j,j'}\|_1 \leq \|\Ehat^{(n-1)}-\Ehat^{\infty}\| \nonumber
\ee
and hence
\be
\|\tilde{\sigma}_{l}-\tilde{\phi}^{n}_{l}\|_{1} \leq \Bigl(\sum_{i=1}^{b}
\sum_{j=1}^b |U_{l,(i j)}| \, \|\chi_{i}\|\Bigl)^{2}\,\|\Ehat^{(n-1)}-\Ehat^{\infty}\| \nonumber
\ee
Since $\sum_{i,j=1}^{b} |U_{l,(i j)}|^{2}\, \|\chi_{i}\|^{2} = \beta_{l}$, two applications of Cauchy-Schwarz give:
\be\label{density:2-norm-final}
\|\tilde{\sigma}_{l}-\tilde{\phi}^{n}_{l}\|_{1} \leq b^2 \beta_{l} \|\Ehat^{(n-1)}-\Ehat^{\infty}\|.
\ee

Finally, combining~(\ref{Ehat:bound}) with~(\ref{density:2-norm-final}), equation~(\ref{tau:first-bound}) becomes:
\be\label{tau:final-bound}
\tau_{l}^{n} \leq c_{1} \lambda^{n}, \quad c_{1} = c b^2 / \lambda.
\ee

To conclude the proof, we note that since the bound for $\tau_{l}^{n}$ is independent of $l$, summing over $l$ in equation~(\ref{epsilon:bound}) yields:
\be
\sum_{l=1}^{L} |\epsilon_{l}^{n}| \leq (\log d^{3}+4) c_{1}\lambda^{n} + \eta(2 c_{1}\lambda^{n}).
\nonumber
\ee

It is clear that for $\lambda' > \lambda$ there exists a constant $c_{2}$ such that 
$$\eta(2 c_{1}\lambda^{n}) \leq c_{2}(\lambda')^{n}.$$
The only condition on $n$ was imposed in equation~(\ref{epsilon:bound})
were we assumed that $2\tau_{l}^{n} \leq \frac{1}{e}$. Using equation~(\ref{tau:final-bound}) we see 
that there is an $n_{0}$ such that the above condition is satisfied for all $n \geq n_{0}$.
The previous observations imply that for all $\lambda'$ with $\lambda < \lambda' < 1$, 
there is a constant $c_{3}$ such that:
\be
\epsilon(n) = c_{3} (\lambda ')^{n} \geq \sum_{l=1}^{L} |\epsilon_{l}^{n}|,\text{  for all $n$.}\nonumber
\ee
Finally, equation~(\ref{eq:main1}) implies that:
\be
\eof{\A}{\B \otimes \B}{\rhoAB} - \eof{\A}{\A^{\otimes n-1}}{\rhoN} \leq \epsilon(n)\nonumber,
\ee
and this completes the proof of the theorem.

A natural question to ask at this point is the following: 
\emph{How does the entanglement between the spin at site $1$
and spins at sites $[p,n],\, p\geq 2$ behave as $p$ becomes large?}
Since the state $\rhoPN$ factorizes into $\rho_{1}\otimes \rho_{[p,n]}$
as $p \rightarrow \infty$~\cite{HP}, we expect that the bulk of the entanglement  is
concentrated near site $1$.
The following theorem confirms this:
\begin{theorem}
For any pure translation invariant \FCS{} and $n \geq p \geq 2$, the following bound holds:
\be
\eof{\A}{\A^{\otimes n-p+1}}{\rhoPN} \leq C \, \ln d\, (n-p+1)\, \epsilon(p),
\ee
where $C$ is of order unity, $d$ is the dimension of each spin and $\epsilon(p)$ decays exponentially fast in $p$.
\end{theorem}
Note that the above theorem implies exponential decay of the entanglement between spin $1$ and
spins $[p,n]$ as long as $n-p$ does not grow exponentially in $p$.
\begin{proof}
The main observation is that the trace distance between the states $\rhoPN$ and $\rho_1\otimes\rho_{[p,n]}$ vanishes exponentially fast with $p$. This is a consequence of the exponential rate of
convergence described in equation~(\ref{Ehat:bound}). To see this, note that $\|\rhoPN - \rho_1\otimes\rho_{[p,n]}\|_1 = \Tr [(\rhoPN - \rho_1\otimes\rho_{[p,n]})\, P]$, where $P$ is the projection onto the positive eigenvalues of $\rhoPN - \rho_1\otimes\rho_{[p,n]}$. But,
\benn
\Tr (\rhoPN\, P) = \sum_{i_1,i_p, i_{p+1},\dots,i_n=1}^d \Tr \left(\rhoN\, P_{i_1}\otimes \one_{[2,p-1]} \otimes P_{i_p}\otimes \cdots \otimes P_{i_n}\right), 
\eenn
and applying (\ref{def:rhoS}) we get,
\be\label{trace:rhoPN}
\Tr (\rhoPN\, P) = \sum_{i_1,i_p, i_{p+1},\dots,i_n=1}^d {\Tr}_{\B} \left(\rho\, \E(P_{i_1} \otimes \Ehat^{p-2}\left(\, \E(P_{i_p}\otimes \cdots \otimes \E(P_{i_n}\otimes \one_{\B})\right)\cdots\right),
\ee
where $P = \sum_{i_1,i_p, i_{p+1},\dots,i_n=1}^d P_{i_1} \otimes P_{i_p}\otimes \cdots \otimes P_{i_n}$ is a decomposition of $P$ into simple tensor products over some fixed basis for each tensor.
Furthermore, from the definition of $\Ehat^{\infty}$ and (\ref{def:rhoS}) we have
\begin{equation}\label{trace:tensor}
\Tr \rho_1\otimes\rho_{[p,n]}\, P = \sum_{i_1,i_p, i_{p+1},\dots,i_n=1}^d {\Tr}_{\B} \left(\rho\, \E(P_{i_1} \otimes \Ehat^{\infty}\left(\, \E(P_{i_p}\otimes \cdots \otimes \E(P_{i_n}\otimes \one_{\B})\right)\cdots\right).
\end{equation}
Combining (\ref{trace:rhoPN}) and (\ref{trace:tensor}) we get that $\|\rhoPN - \rho_1\otimes\rho_{[p,n]}\|_1$
\begin{eqnarray}
 &=& \sum_{i_1,i_p, i_{p+1},\dots,i_n=1}^d {\Tr}_{\B} \left(\rho\, \E(P_{i_1} \otimes (\Ehat^{p-2}-\Ehat^{\infty})\left(\, \E(P_{i_p}\otimes \cdots \otimes \E(P_{i_n}\otimes \one_{\B})\right)\cdots\right)\nonumber \\
&\le& \|\one_{\A} \otimes (\Ehat^{p-2}-\Ehat^{\infty})\|_{\infty}\, \left|\sum {\Tr}_{\B} \left(\rho\, \E(P_{i_1} \otimes \left(\, \E(P_{i_p}\otimes \cdots \otimes \E(P_{i_n}\otimes \one_{\B})\right)\cdots\right)\right|\nonumber\\
&\le& c \lambda^{p-2} \Tr\left(\rho_{[1, n-p+2]} P\right) \le c \lambda^{p-2}
\end{eqnarray}
Now, note that for two density matrices $\rho$ and $\sigma$ the (normalized) trace distance $T(\rho,\sigma) = \frac{\|\rho-\sigma\|_1}{2}$ is an upper bound on the {\it Bures distance} $D(\rho,\sigma) = 2\sqrt{1-F(\rho,\sigma)}$ \cite{Bures}, where $F(\rho,\sigma) = \Tr\sqrt{\rho^{1/2} \sigma \rho^{1/2}}$ is the {\it fidelity} measure. In particular, the following bound holds \cite[Ch. 9]{nielsen_chuang}:
\benn
D(\rho,\sigma)\le 2\sqrt{T(\rho,\sigma)}.
\eenn
Since $\eof{\A}{\A^{\otimes n-p+1}}{\rho_1\otimes\rho_{[p,n]}} = 0$,
a straightforward application of Nielsen's inequality for the continuity of the \EoF~\cite{Nielsen}
yields the desired result.
\end{proof}

\section{Discussion}
Having established such a strong connection between the states $\rhoN$ and $\rhoAB$, one can
apply various entanglement criteria on $\rhoAB$ to deduce entanglement properties of the spin chain.
To start with, we note that for qubit chains with $2$-dimensional memory algebra $\B$, the entanglement
of $\rhoN$ can be computed analytically (in the limit) by evaluating the concurrence~\cite{Wooters} of
$\rhoAB$. For higher dimensions one can apply the PPT criterion to 
$\rhoAB$ to detect distributed entanglement in the finitely correlated state. Specifically, the main theorem in~\cite{Horodecki} implies that there can be no PPT bound entanglement in $\rhoAB$ since
rank$(\rhoAB)=b \leq \max\{d,b\}$. Hence, if the partial transpose of $\rhoAB$ is positive, then $\rhoAB$
is separable. On the other hand, if the partial transpose of $\rhoAB$ is negative, then for $n$ large
enough $\rhoN$ becomes entangled. The amount of maximum entanglement in $\rhoN$ depends on
the amount of entanglement found in $\rhoAB$. From this point of view, it would be very interesting to
look at \FCS{} that maximize entanglement of $\rhoAB$. Moreover, understanding how entanglement of $\rhoAB$ varies with different CP maps $\mathbb{E}$ could lead to a better understanding of how phase
transitions occur when we vary the parameters in the underlying Hamiltonian of the system.

To conclude, we note that the conjecture of Benatti, \etal~\cite{BHN}, 
follows as a corollary of Theorem~\ref{thm:conjecture}. In particular, our result implies that a spin at site $1$ of the chain is entangled
with spins at sites $[2,n]$ (for $n$ large enough) if and only if $\rhoAB$ is entangled. Moreover,
the entanglement of $\rhoN$ approaches the entanglement of $\rhoAB$ exponentially fast.


    %
    %

    \newchapter{Entanglement in the Ground State of Gapped Hamiltonians}{Entanglement in the Ground State of Gapped Hamiltonians}{Entanglement in the Ground State of Gapped Hamiltonians}
    \label{sec:LabelForChapter3}

        
    
       \section{Introduction}
It is widely believed that the entanglement entropy of a region $A$ in the ground state
of a gapped quantum spin Hamiltonian with finite range interactions does not grow
faster than the area of the boundary of $A$. Although this property in its general 
formulation is still a conjecture, this is called the {\em Area Law} for the entropy.
The Area Law is of interest not only for theoretical reasons, but also because it has 
practical implications for the computational complexity of calculating the ground 
state with a desired level of accuracy \cite{vc}. The intuition behind the Area 
Law is simple. By the Exponential Clustering Theorem, a non-vanishing spectral gap 
implies a finite correlation length \cite{ns,hk} and this puts an
exponentially decaying bound on the entanglement of two spins as a function of the distance.
From this, one may guess that only spins near the boundary contribute significantly to
the total entanglement with the exterior of the region and hence, an area law should
hold. The relationship between the correlations and entanglement of a region with its
exterior, however, is not sufficiently well understood to lead to a proof based directly 
on this intuition. If one assumes a decay property of the mutual information instead of
correlations, one can indeed prove an area law \cite{wolf2008}. The Area Law itself
has only been proven in one dimension by Hastings \cite{hastings_law}. In more than one dimension,
for special systems with valence bond (i.e. matrix product) ground states, the area law 
is easy to derive, but there is no general result. Here, we review in detail the one dimensional
result by Hastings and provide a higher dimensional generalization of his ground state approximation.
Although the Area Law bound we discuss below applies only to $1-$dimensional systems, 
we describe the result in a more general setup in case this leads to future insight on how
to treat the Area Law question in higher dimensions.
\section{Setup and Main Results}
We will consider a finite system of spins located at sites in $V$, a finite
subset of $\Ir^d$. At each $x\in V$, we have a finite-dimensional
Hilbert space of dimension $n_x$ and assume that $n_x \le N, \, \forall x\in V$, 
for some constant $N$. Let $\A_X$ be the algebra of
observables associated with $X\subset V$. We consider the dynamics
$\tau_t$ generated by the Hamiltonian
$$
H_V=\sum_{X\subset V} \Phi(X).
$$  The main assumptions are as follows:
The interactions $\Phi(X)$ are uniformly bounded and finite range,
and, for convenience, we will take pair interactions with range $1$. So, for all
$X\subset V$, $\Phi(X) = \Phi(X)^*\in \A_X$, $\| \Phi_X \| \leq J$,
for some constant $J>0$, and $\Phi(X) = 0$ if $\mbox{diam} (X) >1$. Here, for
$X\subset V$,
$$
\mbox{diam} (X) = \max \{ d(x,y) \mid x,y\in X\}.
$$

We will also assume that the Hamiltonian $H_V$ has a unique,
normalized ground state, which we will denote by $\ket{\Psi_0}$, and a
spectral gap $\gamma >0$ to the first excited state. 

For a set $Y \subset V$, the boundary of $Y$, denoted by $\partial
Y$, is
$$
\partial Y = \{ x\in Y \mid \exists y\in V\setminus Y,  
\mbox{with } d(x,y) \leq 1\}.
$$
Moreover, for $\ell \geq 1$, we define the following sets:
\begin{equation}
I_Y = I_Y(\ell) = \left\{ x\in Y\mid \forall y\in \partial Y,
d(x,y) > \ell \right\}.
\end{equation}
The set $I_Y$ corresponds to the $\ell$-interior of $Y$ and it will be
empty if $\ell > \mbox{diam}(Y)$.
\begin{equation}
B_Y = B_Y(\ell) = \left\{ x\in V\mid \exists y\in
\partial Y, d(x,y) \le \ell \right\}.
\end{equation}
The set $B_Y$ corresponds to the $\ell$-boundary of $Y$.
\begin{equation}
E_Y = E_Y(\ell) = \left\{ x\in V\setminus Y\mid \forall y\in
\partial X, d(x,y) > \ell \right\}.
\end{equation}
The set $E_Y$ corresponds to the $\ell$-exterior of $Y$. 
Note that $I_Y,B_Y$ and $E_Y$ are disjoint and moreover, $V=I_Y\cup B_Y\cup E_Y$, as can
be seen from Figure~\ref{fig:area}.

\begin{figure}[htp]
\centering
\includegraphics[width=108mm]{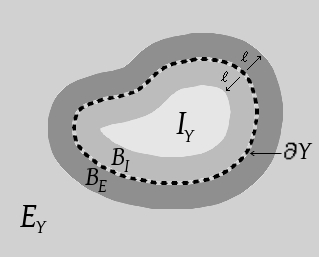}
\caption{The different regions around the boundary of $A$.}\label{fig:area}
\end{figure}

Finally, define $B_I = B_I(\ell) = Y\cap B_Y(\ell)$ and $B_E = B_E(\ell) = (V\setminus Y)\cap B_Y(\ell)$ 
and set $D_I = \Pi_{x\in B_I} \, n_x$ and $D_E = \Pi_{x\in B_E} \, n_x$. 
We will assume from now on that the volume of the interior of the $\ell$-boundary satisfies 
$$|B_I| \le r \, \ell \, |\partial Y|,$$ for some constant $r \ge 1$. For convex $Y$, this
inequality holds with $r=1$. The regions we will be working with will be convex.

We denote by $P_0$ the orthogonal projection onto $\ket{\Psi_0}$ and
by $\rho_A$ the density matrix describing the ground state restricted
to the region $A$: $\rho_A = \Tr_{V\setminus A} P_0$.

We are now ready to state the main result of \cite{hastings_law}.
\begin{theorem}[Area Law]\label{thm:area_law}
Let $H_V$ denote a $1$-dimensional Hamiltonian on an interval $[1,|V|]$, with the properties described above. For any $1\le M \le |V|$, let $A = [1,M]$. Then, the following bound, independent of $M$ applies to the entropy of the region $A$:
\be\label{eq:area_law}
S(\rho_A) \leq C\,\xi \ln C(\gamma,d,J) \ln N \, 2^{2\xi \ln N}.
\ee
where $C$ is of order unity and the constants $\xi$ and $C(\gamma,d,J)$ are defined in Theorem \ref{thm:gs_approx}.
\end{theorem}

We will provide the complete proof of Theorem \ref{thm:area_law} assuming Theorem 
\ref{thm:gs_approx} below, which is a generalization to higher dimensions of a result of 
Hastings \cite{hastings_law}. The proof of Theorem \ref{thm:gs_approx} proceeds along the 
same lines as in the one-dimensional case \cite{hastings_law, qad,loc},
but it is a bit long and technical and therefore we present its proof in section \ref{sec:LabelForChapter3:Section2}. 

\begin{theorem}[Ground state approximation]\label{thm:gs_approx}
There exists $\xi > 0$, such that for any sufficiently large $\ell \ge c_0\, d^2 \xi^2$ and $Y \subset V$,
there exist two orthogonal projections $\PI \in \A_Y$, 
$\PE \in \A_{V\setminus Y}$ and an operator 
$\PB \in \A_{B_Y(\ell)}$ with $\|P_B\| \le 1$, such that
\begin{equation} \label{eq:main2}
\| \PB\PI\PE - P_0 \| \leq C(\gamma,d,J) | \partial Y |^2 e^{- \ell/ \xi}
\end{equation}
where $P_0$ is the projection onto the (unique) ground state, $c_0$ is of order unity and $C(\gamma,d,J)$, $\xi$ are explicit in terms of $J$, $\gamma$ and $d$, the dimensionality of $V$.
\end{theorem}

The operator $\PB$ in (\ref{eq:main2}) is responsible for all correlations in
the ground state approximation between the region $Y$ and its exterior. Its
support is concentrated along the boundary of $Y$. This is reminiscent of the 
structure of matrix product states \cite{AKLT,FNW2,VC}. The finite extent
of correlations and entanglement across any boundary is essentially a consequence
of the non-vanishing spectral gap and the existence of a finite Lieb-Robinson
velocity \cite{lr,ns,hk,loc_est}. The problem of calculating such an approximation
of the ground state is a related but separate question we intend to turn to at
a later occasion. A numerical algorithm for this problem is discussed in 
\cite{solvgap}. It is known that in general this is an NP-hard problem \cite{npmp}.

\section{Proof of Main Theorem}
We present now the proof of Theorem~\ref{thm:area_law}.

\begin{proof}
First, we introduce the length $m^* = \frac{S(\rho_A)}{4\ln N |\partial A|}$ and then set $Y \equiv Y(m) \equiv I_A(m), \, m \le m^*$. The first step is to use the well known subadditivity of the Von Neumann entropy, which follows from the non-negativity of quantum relative entropy $S(\rho\|\sigma) = \Tr (\rho \ln \rho - \rho \ln \sigma)$, to get the following bound for $S(\rho_A)$:
\begin{equation}
S(\rho_A \| \rho_Y \otimes \rho_{A\setminus Y}) = S(\rho_Y) + S(\rho_{A\setminus Y}) - S(\rho_A) \ge 0
\end{equation}
The above bound, combined with the fact that 
\benn
S(\rho_{A\setminus Y}) \le |A\setminus Y| \ln N \le m^* |\partial A| \ln N = S(\rho_A)/4,
\eenn
implies that
\begin{equation}\label{entropy:bound}
S(\rho_A) \le \frac{4}{3}\, S(\rho_{Y(m)}), \qquad m \le m^*
\end{equation}
We can now focus on bounding the entropy of the region $Y(m)$. 
We will do this by looking at the rate of
decay of the Schmidt coefficients in the Schmidt decomposition of $\ket{\Psi_0}$ along the boundary of $Y(m)$. Since we will be using the Schmidt decomposition of the ground state $\ket{\Psi_0}$, let us introduce it here as $$\ket{\Psi_0} = \sum_{\alpha}
\sqrt{\sigma_0(\alpha)} \ket{\Psi_{{Y},0}(\alpha)}\otimes \ket{\Psi_{{V\setminus Y},0}(\alpha)},$$ where $\sum_{\alpha} \sigma_0(\alpha) = 1$ and $\{\ket{\Psi_{Y,0}(\alpha)}\}$, $\{\ket{\Psi_{V\setminus Y,0}(\alpha)}\}$ are orthonormal sets supported on $Y$ and $V\setminus Y$, respectively. We order the Schmidt coefficients of $\ket{\Psi_0}$ in decreasing order such that if $\alpha<\beta$ then $\sigma_0(\alpha) \geq \sigma_0(\beta)$. Moreover, 
$$\rho_Y = \sum_\alpha \sigma_0(\alpha) \pure{\Psi_{Y,0}(\alpha)}$$ and 
$$\rho_{V\setminus Y} = \sum_\beta \sigma_0(\beta) \pure{\Psi_{V\setminus Y,0}(\beta)}.$$
Having fixed $m^*$ and $Y(m)$, we are interested in the overlap of the density matrix $\rho_Y \otimes \rho_{V\setminus Y}$ with $P_0$, since we will be treating the cases when the overlap is small and when it is large, separately. We define the overlap to be 
\be\label{overlap}
P \equiv P(m) \equiv \Tr (P_0\, \rho_Y \otimes \rho_{V\setminus Y}).
\ee
We use the approximation operator $\PB\PI\PE$ of Theorem \ref{thm:gs_approx} in order to relate the overlap $P$ with the approximation error $\epsilon(\ell)$, over
which we have some control. Thus, we choose the following approximation to the ground state:
\begin{equation}
\rho(\ell) \equiv \PB\PI\PE \rho_Y \otimes \rho_{V\setminus Y} \PI\PE\PB^{\dagger} 
\end{equation} 
which, as we will see shortly, has Schmidt rank bounded by an exponential in $m$ and $|\partial Y|$, but independent of the volume of $Y$. More concretely, we will show that $\rho(\ell)$ has a decomposition into pure states, each with Schmidt decomposition along the boundary $\partial Y$ with Schmidt rank at most $N^{2\,\ell \, |\partial Y|}$. To see this, first note that pure states in the decomposition of $\PI\PE \rho_Y \otimes \rho_{V\setminus Y}\PI\PE$ are product states along $\partial Y$. Using the spectral decompositions of $\rho_Y$ and $\rho_{V\setminus Y}$ introduced earlier, we may focus our attention to product states of the form 
\be
\ket{\Psi_Y(\gamma,\delta)} = \PI \ket{\Psi_{Y,0}(\gamma)} \otimes \PE \ket{\Psi_{V\setminus Y,0}(\delta)}.
\ee
We study now how the action of $\PB$ on $\ket{\Psi_Y(\gamma,\delta)}$ affects its Schmidt rank.

Since $\PB$ is an operator acting non-trivially only on sites in a subset of $B_Y(\ell)$, we have the following general decomposition:
\begin{equation}\label{decomp}
\PB = \sum_{\alpha,\beta=1}^{D_I} \one_{I_Y} \otimes E(\alpha,\beta) \otimes G(\alpha,\beta) \otimes \one_{E_Y},
\end{equation}
where the $D_E \times D_E$ matrices $G(\alpha,\beta)$ act on sites in $B_E$ and the matrix units $E(\alpha,\beta)$, which act non-trivially on $B_I$, form an orthonormal basis for $D_I \times D_I$ matrices. Moreover, $D_I \le N^{|B_I|} \le N^{\ell |\partial Y|}$.

To bound the Schmidt rank of $\PB \ket{\Psi_Y(\gamma,\delta)}$ we trace over sites in $Y$ and study the rank of the operator
\begin{equation}\label{Schmidt_bound}
{\Tr}_{Y} \left(\PB \pure{\Psi_Y(\gamma,\delta)} \PB^{\dagger}\right) =\sum_{\alpha=1}^{D_I}\left[ \sum_{\beta,\beta'=1}^{D_I}
c(\beta,\beta')\, \ketbra{F_{\alpha}(\beta)}{F_{\alpha}(\beta')}\right],
\end{equation}
with $\ket{F_{\alpha}(\beta)}  = G(\alpha,\beta)\otimes \one_{E_Y} \PE \ket{\Psi_{V\setminus Y,0}(\delta)}$ and 
$c(\beta,\beta') = \bra{\Psi_{Y,0}(\gamma)} \PI \one_{I_Y} \otimes E(\beta',\beta) \PI \ket{\Psi_{Y,0}(\gamma)}$, coming from (\ref{decomp}) and the definition of $\ket{\Psi_Y(\gamma,\delta)}$.

Clearly, as a sum of $D_I$ matrices each with rank at most $D_I$, the above operator has rank bounded by $D_I^2$ and, hence, the Schmidt rank of $\PB\ket{\Psi_Y(\gamma,\delta)}$ is bounded above by $N^{2 \ell \,|\partial Y|}$.

Now that we have a good grasp on the entanglement overhead produced by the approximation operator $\PB\PI\PE$, we return to the question of how the overlap $P$ relates to the approximation error $\epsilon(\ell)$.
Remembering from (\ref{overlap}) that $P = \Tr(P_0 \rho_Y\otimes \rho_{V\setminus Y})$, we have the following overlap estimate between $\rho(\ell)$ and $P_0$:
\begin{eqnarray*}
\sqrt{\Tr(P_0 \rho(\ell))} &=& \|P_0 \PB\PI\PE \sqrt{\rho_Y} \otimes \sqrt{\rho_{V\setminus Y}}\|_2\\
 &\ge&
\|P_0 \sqrt{\rho_Y} \otimes \sqrt{\rho_{V\setminus Y}}\|_2 - \|P_0(P_0 - \PB\PI\PE)\|\| \sqrt{\rho_Y} \otimes \sqrt{\rho_{V\setminus Y}}\|_2 \\
&\ge& \sqrt{P}-\epsilon(\ell),
\end{eqnarray*}
where $\|X\|_2 = \sqrt{\Tr XX^\dagger}$ and $\|\cdot\|$ is the sup-norm.
Hence, for all $\ell$ such that $\epsilon^2(\ell) \le P$, we have 
\begin{equation}
\Tr(P_0 \rho(\ell)) \ge \left(\sqrt{P}-\epsilon(\ell)\right)^2.\label{hg}
\end{equation}
Moreover, we have that
\begin{eqnarray}
\Tr [(1-P_0) \rho(\ell)] &=& \| (1-P_0)(\PB\PI\PE - P_0)\sqrt{\rho_Y} \otimes \sqrt{\rho_{V\setminus Y}}\|_2^2\nonumber\\
&\le&\|1-P_0 \|^2 \|\PB\PI\PE - P_0 \|^2 \le \epsilon^2(\ell)
\end{eqnarray}
Finally, upon normalization the overlap becomes:
\begin{eqnarray}
\frac{\Tr(P_0 \rho(\ell))}{\Tr (\rho(\ell))} = 1- \frac{1}{1+\frac{\Tr(P_0 \rho(\ell))}{\Tr[(1-P_0)\rho(\ell)]}}
\ge 1- \frac{1}{1+\frac{\left(\sqrt{P} - \epsilon(\ell)\right)^2}{\epsilon^2(\ell)}}
\ge 1-2\frac{\epsilon^2(\ell)}{P}\label{P:large}\qquad
\end{eqnarray}
We will now use (\ref{P:large}) for the case when $P \ge 2\epsilon^2(m^*)$.
The case $P \le 2\epsilon^2(m^*)$ will be treated later on.
$$\mbox{{\bf Case I}:  } P \ge 2\epsilon^2(m^*).$$
In this case, the next step is to relate the above overlap to the Schmidt coefficients of $\ket{\Psi_0}$.
More specifically, we will now show that for $\ell \ge 0$:
\be
\label{constraint1}
\sum_{\alpha \le N^{2 \ell |\partial Y|}} \sigma_0(\alpha) \geq \frac{\Tr(P_0 \rho(\ell))}{\Tr (\rho(\ell))}.
\ee
To prove this, first note we already showed that $\rho(\ell)$ is a convex combination of pure states $\PB\ket{\Psi_Y(\gamma,\delta)}$, each with Schmidt rank bounded above by $N^{2\ell |\partial Y|}$.  Let $\ket{\Psi_0(\ell)}$ be the (not necessarily unique) pure state in the aforementioned decomposition of $\rho(\ell)$ satisfying 
\be
\Tr\left(P_0 \pure{\Psi_0(\ell)}\right) \ge \frac{\Tr(P_0 \rho(\ell))}{\Tr (\rho(\ell))}
\ee
and introduce its Schmidt decomposition as 
\be
\sum_{\beta = 1}^s \sqrt{\tau_{\beta}} \ket{\Phi_{Y}(\beta)} \otimes \ket{\Phi_{V\setminus Y}(\beta)}, \mbox{ with } \sum_{\beta=1}^s \tau_{\beta} = 1 \mbox{ and } s \le N^{2\ell|\partial Y|},
\ee 
as we have already demonstrated. For notational convenience, let 
$M_{Y}(\alpha,\beta) = \left|\braket{\Phi_{Y}(\beta)}{\Psi_{{Y},0}(\alpha)}\right|$ and $M_{V\setminus Y}(\alpha,\beta) = \left|\braket{\Phi_{V\setminus Y}(\beta)}{\Psi_{{V\setminus Y},0}(\alpha)}\right|.$
Note that since each of $\{\ket{\Phi_{Y}(\beta)}\}$, $\{\ket{\Phi_{V\setminus Y}(\beta)}\}$, $\{\ket{\Psi_{Y,0}(\alpha)}\}$ and $\{\ket{\Psi_{V\setminus Y,0}(\alpha)}\}$ is an orthonormal set, Bessel's inequality implies that $$\sum_{\beta=1}^s M_Y(\alpha,\beta)^2 \le 1 \mbox{ and } \, \sum_{\alpha} M_{V\setminus Y}(\alpha,\beta)^2 \le 1,$$ as well as 
$$\sum_{\alpha} M_Y(\alpha,\beta)^2 \le 1 \implies \sum_{\alpha, \beta} M_Y(\alpha,\beta)^2 \le s.$$

Then, an application of the triangle inequality followed by Cauchy-Schwarz gives the following upper bound for $\Tr(P_0 \pure{\Psi_0(\ell)}) = |\braket{\Psi_0(\ell)}{\Psi_{0}}|^2:$
\begin{eqnarray*}
|\braket{\Psi_0(\ell)}{\Psi_{0}}|^2 &\le&\left(\sum_{\alpha,\beta} \sqrt{\sigma_0(\alpha)}\sqrt{\tau(\beta)}\, M_{Y}(\alpha,\beta) M_{V\setminus Y}(\alpha,\beta) \right)^2 \\
&\le&\left(\sum_{\alpha,\beta}\sigma_0(\alpha) M_{Y}(\alpha,\beta)^2 \right) \, \left(\sum_{\alpha,\beta} \tau(\beta) M_{V\setminus Y}(\alpha,\beta)^2 \right) \\
&\le&\sum_{\alpha \le s} \sigma_0(\alpha).
\end{eqnarray*}
The last inequality follows from Schur convexity of $f([p(\alpha)]) = \sum_{\alpha} \sigma_0(\alpha) \, p(\alpha)$ and the observation that the vector $[1,1,\ldots,1,0,\ldots,0]$, with at most $s$ ones, majorizes $\left[\sum_{\beta} M_Y(1,\beta)^2,\sum_{\beta} M_Y(2,\beta)^2,\ldots\right]$. 

To see that $f([p(\alpha)])$ is Schur convex, note that if we set
$S_p(\alpha) = \sum_{k=1}^{\alpha} p(k)$ and $\Delta(\alpha,\beta) = \sigma_0(\alpha)-\sigma_0(\beta)$ then the condition that $[p(\alpha)]$ majorizes $[q(\alpha)]$ ($p \succeq q$) becomes $p \succeq q \Leftrightarrow S_p(\alpha) \ge S_q(\alpha),\, \forall \alpha$. Moreover, 
$$f(\{p(\alpha)\})-f(\{q(\alpha)\}) = \sum_{\alpha} \Delta(\alpha,\alpha+1) \left(S_p(\alpha)-S_q(\alpha)\right),$$
which is non-negative since we have arranged the $\sigma_0(\alpha)$ in decreasing order so that
$\Delta(\alpha,\alpha+1) \ge 0, \, \forall \alpha$.

Now that we have demonstrated (\ref{constraint1}), we may use it in combination with (\ref{P:large}) to
show that $S(\rho_Y)$ satisfies an area law in Case I.
We begin by setting $\ell'$ to be the smallest integer $\ell$ such that $2 \epsilon^2(\ell)/P \leq 1$.
Using (\ref{P:large}) and (\ref{constraint1}), we have for $\ell>\ell'$
\be
\label{constraint2}
\sum_{\alpha\geq N^{2\ell \, |\partial Y|}+1} \sigma_0(\alpha)\leq
\exp[-2(\ell-\ell')/\xi].
\ee

We now maximize the entropy of $\rho_Y = {\Tr}_{V\setminus Y} P_0$, given by
$$S(\rho_Y)=-\sum_{\alpha=1} \sigma_0(\alpha) \ln(\sigma_0(\alpha))$$
subject to the constraint (\ref{constraint2}).
Following the notation of Lemma~\ref{lem:ent_bound}, set $s_n = N^{2\,(\ell'+n)|\partial Y|}, \, n \ge 1$ and
$c = \exp[-2/\xi]$. Then, $R = N^{2\,|\partial Y|}$ and $s_1 = N^{2\,(\ell'+1)|\partial Y|}$ and Lemma~\ref{lem:ent_bound} implies that $S(\rho_Y)$ is bounded above by:
\begin{equation}\label{CaseI:bound}
\left(\xi+ 2 + 2 \ell' \right) \ln N\, |\partial Y|\, + \left(\frac{\xi}{2}+1\right) \ln 2,
\end{equation}
where we used the inequalities $2c/(1-c) \le \xi$, $1/(1-c) \le \xi/2 + 1$ and $H_2(1-c) \le \ln 2$.
Using the definition of $\ell'$ we have that $$2 \ell' \le \xi \ln \left(\frac{2 C^2(\gamma,d,J) |\partial Y|^4}{P}\right) + 2 = \xi \ln \left(\frac{2\epsilon^2(m^*)}{P}\right) + 2m^* + 2.$$
Since we are considering Case I, we have furthermore that $2 \ell' \le 2m^* + 2$.
Hence, using (\ref{CaseI:bound}) we get
\be
S (\rho_Y) \le C_1 \ln N\, |\partial Y|\, + S(\rho_A)/2 + C_2
\ee
with $C_1 = \xi + 4$ and $C_2 =  \left(\frac{\xi}{2}+1\right) \ln 2$. Using (\ref{entropy:bound}) and rearranging terms gives
\be
S(\rho_A) \le 4 C_1 \ln N\, |\partial Y|\, + 4 C_2 \le 4 (C_1 \ln N+C_2) \, |\partial A|,
\ee
which is an area law bound.

It remains to treat the following case:
$$\mbox{{\bf Case II}:  } P \le 2\epsilon^2(m^*).$$
The main idea here is to show that in this case, there is a non-trivial lower bound on the relative entropy
\be\label{rel_ent:bound0}
S_R(\ell) \equiv S\left(\rho_{B(\ell)} \| \rho_{B_I(\ell)}\otimes \rho_{B_E(\ell)}\right) = S(\rho_{B_I(\ell)}) + S(\rho_{B_E(\ell)}) - S(\rho_{B(\ell)}),
\ee
remembering that $B(\ell) = B_Y(\ell)$ is the $\ell-$boundary of $Y(m) = I_A(m),$ for $m \le m^*$, and that $B_I(\ell), B_E(\ell)$ represent the interior and exterior of $B(\ell)$ along $\partial Y$, respectively. In particular, we will show that for any fixed $m$ and for $\ell$ sufficiently large,
\begin{equation}\label{rel_ent:bound1}
S_R(\ell) \ge -(1-\epsilon(\ell))^2 \ln (2P + 4\epsilon(\ell)) + (1-\epsilon(\ell))^2\ln (1-\epsilon(\ell))^2 -\ln2,
\end{equation}
which, in the case that $\epsilon(m^*) \le 1/2$ and $ P \le 2\epsilon^2(m^*)$, implies 
\begin{eqnarray}
S_R(\ell) &\ge& -(1-\epsilon(\ell))^2 \ln (4\epsilon^2(m^*) + 4\epsilon(\ell)) + (1-\epsilon(\ell))^2\ln (1-\epsilon(\ell))^2 -\ln2\nonumber\\
 &\ge& -(1-\epsilon(\ell))^2\ln (2\epsilon(m^*) + 4\epsilon(\ell)) +(1-\epsilon(\ell))^2\ln (1-\epsilon(\ell))^2-\ln2 \nonumber\\
 &\ge& -(1-\epsilon(\ell))^2\ln (6\epsilon(\ell))+(1-\epsilon(\ell))^2\ln (1-\epsilon(\ell))^2 -\ln2\nonumber\\
 &\ge& -\ln \epsilon(\ell) -\ln 12, \qquad \ell \le m^*,\label{rel_ent:bound2}
\end{eqnarray}
where we used $(1-\epsilon(\ell))^2\ln (1-\epsilon(\ell))^2+(1-(1-\epsilon(\ell))^2)\ln 2\epsilon(\ell) \ge -H_2\left((1-\epsilon(\ell))^2\right)\ge -\ln 2$ in the final inequality.
Note that if $\epsilon(m^*) \ge 1/2$, then we have the following upper bound on $m^*$:
$$m^* \le \xi \ln(2 \,C(\gamma,d,J)\, |\partial Y|^2),$$ which implies an area law bound on $S(\rho_A)$ with a logarithmic correction:
$$S(\rho_A) \le (4\xi \ln (2 C(\gamma,d,J)) + 8\xi \ln |\partial A|) \ln N |\partial A|,$$
where we have used the definition of $m^*$ and the simple inequality $|\partial Y| \le |\partial A|$.

We now go back to proving (\ref{rel_ent:bound1}), which, combined with (\ref{rel_ent:bound0}) and (\ref{rel_ent:bound2}), ultimately implies for $m \le m^*$ and $\ell_1\le \ell \le m^*$:
\be\label{rel_ent:bound*}
S(\rho_{B_I(\ell)}) + S(\rho_{B_E(\ell)}) - S(\rho_{B(\ell)}) \ge \frac{\ell}{\xi} -\ln (12\, C(\gamma,d,J)\, |\partial Y(m)|^2) = \frac{\ell-\ell_2}{\xi},
\ee
where $\ell_1$ is the smallest $\ell$ for which $\epsilon(\ell) \le 1$ and $\ell_2 = \xi\ln (12\, C(\gamma,d,J)\, |\partial Y(m)|^2) \le \xi\ln12 +\ell_1$.

We begin by restating the following important result about the relative entropy \cite{lindblad,uhlmann,petz}:
\begin{lemma}[{\bf Monotonicity of Relative Entropy}]
Let $\mathcal{H}$ and $\mathcal{K}$ be finite dimensional Hilbert spaces, and let $\rho,\, \sigma$ be two density matrices on $\mathcal{H}$. Then, for any trace preserving, $2$-positive map $\Phi : \mathcal{B}(\mathcal{H}) \rightarrow \mathcal{B}(\mathcal{K})$, the following inequality holds for the relative entropy:
\be\label{rel_ent:mono}
S(\rho\|\sigma) \ge S(\Phi(\rho)\|\Phi(\sigma))
\ee
\end{lemma}

In particular, for any operator $M$ on $\mathcal{H}$, with $\|M\| \le 1$, setting 
\be
\Phi(\rho) = 
\left(
\begin{array}{cc}
\Tr(|M| \rho) & 0\\
0 & \Tr\left((1-|M|) \rho\right)\\
\end{array}
\right),
\ee 
we get a trace preserving, $2$-positive map $\Phi : \mathcal{B}(\mathcal{H}) \rightarrow \mathcal{B}(\mathbb{C}^2)$.
To verify that $\Phi$ is $2$-positive, we note that for any positive $2\times2$ block matrix $A = \left(\begin{array}{cc} a & b\\ b^\dagger&c\end{array}\right) \ge 0$, the determinant of the $4\times 4$ matrix $\one_2\otimes \Phi \left(A\right) = \left(\begin{array}{cc} \Phi(a)&\Phi(b)\\ \Phi(b^\dagger)& \Phi(c)\end{array}\right)$ is given by the product of the determinants of the following two matrices:
\begin{eqnarray*}
&\one_2\otimes \Tr& \left( \left(\begin{array}{cc} \sqrt{|M|} & 0\\ 0 &\sqrt{|M|} \end{array}\right)\left(\begin{array}{cc} a & b\\ b^\dagger&c \end{array}\right)\left(\begin{array}{cc} \sqrt{|M|} & 0\\ 0 &\sqrt{|M|} \end{array}\right) \right)\\
&\one_2\otimes \Tr& \left( \left(\begin{array}{cc}\sqrt{1-|M|} & 0\\ 0 &\sqrt{1-|M|} \end{array}\right)\left(\begin{array}{cc} a & b\\ b^\dagger&c \end{array}\right)\left(\begin{array}{cc} \sqrt{1-|M|} & 0\\ 0 &\sqrt{1-|M|} \end{array}\right) \right)
\end{eqnarray*}
whose positivity follows immediately from the $2$-positivity (in fact, complete positivity) of the trace operator.

For our purposes, we will choose $\rho = \rho_{B(\ell)}$ and $\sigma = \rho_{B_I(\ell)}\otimes\rho_{B_E(\ell)}$ and take $M$ to be $\PB'$, the non-trivial part of $\PB = \PB'\otimes \one_{V\setminus B(\ell)}$, noting that $\|\PB'\| \le 1$, since $\|\PB\| \le 1$. Letting $\mu_1 = \Tr(|\PB'| \rho_{B(\ell)})$ and $\mu_2 = \Tr(|\PB'| \rho_{B_I(\ell)}\otimes \rho_{B_E(\ell)})$, we get from (\ref{rel_ent:mono}) and the definition of $\Phi$:
\be\label{rel_ent:bound4}
S\left(\rho_{B(\ell)} \| \rho_{B_I(\ell)}\otimes \rho_{B_E(\ell)}\right) \ge -H_2(\mu_1) - \mu_1 \ln \mu_2 - (1-\mu_1) \ln (1-\mu_2) \ge -\mu_1 \ln \mu_2 - \ln 2,
\ee
where we used again the bound $H_2(c) \le \ln 2$ for the binary entropy and the observation that $(1-\mu_1) \ln (1-\mu_2) \le 0$. We will now show that $\mu_1 \ge (1-\epsilon(\ell))^2$ and $\mu_2 \le (2P+4\epsilon(\ell))/(1-\epsilon(\ell))^2$. We start with the lower bound for $\mu_1 = \Tr(|\PB'| \rho_{B(\ell)}) = \Tr(|\PB| P_0) \ge \Tr(|\PB|^2 P_0) = \|P_0 \PB\|_2^2$ and hence,
\be
\sqrt{\mu_1} \ge \|P_0 \PB\|_2 \ge \|P_0 \PB\PI\PE\|_2 \ge \|P_0^2\|_2 - \|P_0 (\PB\PI\PE - P_0)\|_2 \ge 1 - \epsilon(\ell),
\ee
which translates, for $\ell\ge \ell_1$ such that $\epsilon(\ell) \le 1$, into the lower bound $\mu_1 \ge (1-\epsilon(\ell))^2$. For $\mu_2$, we note that $\Tr\left(|\PB'| \rho_{B_I(\ell)}\otimes \rho_{B_E(\ell)}\right) = \Tr(|\PB| \rho_{Y}\otimes \rho_{V\setminus Y})$. 
Moreover, given the polar decomposition $\PB = U |\PB|$, we have that 
\begin{eqnarray*}
\Tr(|\PB|\PI\PE \rho_{Y}\otimes \rho_{V\setminus Y}) 
&=& \Tr(U^\dagger \PB\PI\PE \rho_{Y}\otimes \rho_{V\setminus Y})\\ 
&\le& \|U^\dagger\| \, |\Tr(\PB\PI\PE \rho_{Y}\otimes \rho_{V\setminus Y})| \\
&\le& \Tr(P_0\rho_{Y}\otimes \rho_{V\setminus Y})+ \|\PB\PI\PE-P_0\|\\
&\le& P+\epsilon(\ell).
\end{eqnarray*}
Now, observe that setting $\mu_3 = \Tr(\PI\PE \rho_{Y}\otimes \rho_{V\setminus Y})$ we have:
\be
\mu_3 \ge |\Tr(\PI\PE P_0)|^2\ge |\Tr(\PB\PI\PE P_0)|^2 \ge (1-\epsilon(\ell))^2
\ee 
from which we get:
\begin{equation}
\Tr \left[ |\PB| \PI\PE\rho_{Y}\otimes \rho_{V\setminus Y}  \right] = \Tr \left[
\left(|\PB|-\mu_2 \right) \left( \PI\PE -\mu_3 \right) \rho_{Y}\otimes \rho_{V\setminus Y} \right] 
+ \mu_2 \mu_3.
\end{equation}
Hence, by the triangle inequality,
\begin{equation}
\mu_2 \mu_3 \leq \Tr \left[ |\PB| \PI\PE\rho_{Y}\otimes \rho_{V\setminus Y}  \right]
+ \left| \Tr \left[ \left(|\PB|-\mu_2 \right) \left( \PI\PE -\mu_3 \right) \rho_{Y}\otimes \rho_{V\setminus Y} \right]  \right|.
\end{equation}
Moreover, by Cauchy-Schwarz, we have the following upper bound
\begin{eqnarray}
&&\left| \Tr \left[
\left(|\PB|-\mu_2 \right) \left( \PI\PE -\mu_3 \right) \rho_{Y}\otimes \rho_{V\setminus Y} \right] \right|\nonumber\\
&& \leq  \sqrt{\Tr \left[
\left(|\PB|-\mu_2 \right)^2 \rho_{Y}\otimes \rho_{V\setminus Y} \right] } \cdot \sqrt{\Tr \left[
\left(\PI \PE-\mu_3 \right)^2 \rho_{Y}\otimes \rho_{V\setminus Y} \right]}
\nonumber \\
&& =  \sqrt{\Tr \left[ |P_B|^2 \rho_{Y}\otimes \rho_{V\setminus Y} \right] - \mu_2^2 } \cdot \sqrt{\mu_3-\mu_3^2} \nonumber \\
&& \leq  \sqrt{\mu_2 - \mu_2^2} \cdot \sqrt{\mu_3 - \mu_3^2}\le \sqrt{2\epsilon(\ell)} \sqrt{\mu_2 \mu_3}
\end{eqnarray}
where we used $\|\PB\|\le 1$ in the previous to last inequality and $1-\mu_3 \le 2\epsilon(\ell)$ in the last inequality. Setting $x = \sqrt{\mu_2 \mu_3}$, the above inequalities imply:
\be
x^2 - \sqrt{2\epsilon(\ell)} x - (P+\epsilon(\ell)) \le 0,
\ee
and solving the quadratic inequality implies that $x^2 \le 2P + 4\epsilon(\ell),$ and hence,
\be
\mu_2 \le \frac{2P+4\epsilon(\ell)}{(1-\epsilon(\ell))^2}.
\ee
Plugging things back in (\ref{rel_ent:bound4}), yields (\ref{rel_ent:bound1}). We are now ready to make use of bound (\ref{rel_ent:bound*}). Setting $$S_{\ell} = \max \left\{S(\rho_{B_{Y(m)}(\ell)}): B_{Y(m)}(\ell) \subset B_{Y(m^*/2)}(m^*/2)\right\},$$ we get from (\ref{rel_ent:bound*}):
\be
S_{2\ell} \le 2 S_\ell - \frac{\ell-\ell_2}{\xi}
\ee
and one can easily check that for $\ell_2 \le \ell \le m^*/2$ we have:
\be
S_{\ell} \le \ell\, \frac{S_{\ell_2}}{\ell_2} - \ell\, \frac{\lg (\ell/(4\ell_2))}{2\xi} - \frac{\ell_2}{\xi}.
\ee
Finally, since $S_{m^*/2} \ge 0$ and $S_{\ell_2} \le \ell_2 \ln N |\partial A|$, the above bound implies that:
\be
m^* \le 8\ell_2 \, 2^{2\xi \ln N |\partial A|} \implies S(\rho_A) \le 32\,\ell_2 \ln N \, 2^{2\xi \ln N |\partial A|}\, |\partial A|,
\ee
which ultimately yields:
\be
S(\rho_A) \le 32\,\xi(\ln(12C(\gamma,d,J))+2\ln |\partial A|) \ln N \, 2^{2\xi \ln N |\partial A|}\, |\partial A|.
\ee
Obviously, the above bound gives a constant upper bound in one dimension. Nevertheless, the
exponential factor prohibits any direct extension of the area law to higher dimensions.
\end{proof}
\subsection{Proof of Entropy Bound}
In the proof of Case I in the above Theorem the following lemma was used.

\begin{lemma}[{\bf Entropy bound}]\label{lem:ent_bound} Let $\rho$ be a density matrix with spectral decomposition $\sum_{\alpha} \sigma(a) \pure{\psi(\alpha)}$ and assume that there is $0 < c < 1$, such that for an increasing sequence $s_n,\, n \ge 1$ with $s_{n+1}/s_n \le R, \, \forall n\ge 1$ and $s_1 > 1$, the following constraint holds on the eigenvalues $\sigma(\alpha)$:
\be
\sum_{\alpha \ge s_n + 1} \sigma(\alpha) \le c^n, \, \forall n \ge 1.\label{constraint}
\ee
Then, the entropy of $\rho$ is bounded by:
$$S(\rho) \le \ln s_1 + \frac{c}{1-c}\ln R +  \frac{1}{1-c} H_2(1-c),$$
where $H_2(1-c) = -c \ln c - (1-c) \ln (1-c)$ is the binary entropy.
\begin{proof}
Since the Shannon entropy is Schur-concave, it is maximized when the input probability distribution $\{\sigma(\alpha)\}$ is ``minimized per entry'' (i.e. majorized by all other probability vectors satisfying the constraint (\ref{constraint})). Note that $s_{n+1} \le s_1 R^n, \, n \ge 1$ and hence $s_{n+1} - s_n \le s_1 R^n.$

For $\alpha \le s_1$ we have that $-\sum_{\alpha}\sigma(\alpha) \ln \sigma(\alpha)$ is minimized when all $\sigma(\alpha)$ are equal and their sum is minimized according to (\ref{constraint}):
$$\sum_{\alpha=1}^{s_1} \sigma_0(\alpha)= 1-c \implies \sigma(\alpha) = \frac{1-c}{s_1}.$$
Similarly, for $\alpha \in [s_n+1, s_{n+1}], \, n \ge 1,$ we see that the entropy is minimized when
\benn
\sum_{\alpha=s_n+1}^{s_{n+1}} \sigma(\alpha)= c^n-c^{n+1} \implies \sigma(\alpha) = \frac{(1-c)c^n}{s_{n+1}-s_n}.
\eenn 
Gathering terms, the entropy can be written as $S(\rho) = \sum_{n = 0}^{\infty} S_n(\rho)$, where $S_0(\rho) = -\sum_{\alpha = 1}^{s_1}\sigma(\alpha) \ln \sigma(\alpha)$ and $S_n(\rho) = - \sum_{\alpha = s_n+1}^ {s_{n+1}} \sigma(\alpha) \ln \sigma(\alpha),\, n\ge 1$. We can bound each $S_n(\rho)$ as follows:
\begin{equation*}
S_0(\rho) \le -(1-c)\ln(1-c) + (1-c)\ln s_1
\end{equation*}
and for $n\ge 1,$
\begin{eqnarray*}
S_n(\rho) \le (1-c)c^n\Big(-\ln((1-c)c^n) + \ln (s_{n+1}-s_n)\Big)\\
\le c^n(1-c) (\ln s_1 - \ln(1-c)) + nc^{n}(1-c)(\ln R - \ln c)
\end{eqnarray*}
Hence, we get the desired bound:
\begin{eqnarray*}
S(\rho) &\le& \ln s_1-\ln (1-c) + (1-c)\Big(\sum_{n=1}^\infty nc^{n}\Big) (\ln R -\ln c) \\
&=& \ln s_1 + \frac{c}{1-c}\ln R + \frac{1}{1-c} H_2(1-c).
\end{eqnarray*}
and this completes the proof of this Lemma.
\end{proof}
\end{lemma}

        \section[Approximating the Ground State]{Approximation for the ground state of a gapped Hamiltonian}
        \label{sec:LabelForChapter3:Section2}
       In this section, we construct an approximation to the ground state of a gapped Hamiltonian.
The setup is the same as in the previous section, but we restate the main assumptions here
for the sake of completeness.
Again, we consider a system of the following type: Let $V$ be a finite
subset of $\Ir^d$. At each $x\in V$, we have a finite-dimensional
Hilbert space of dimension $n_x$. Let $\A_X$ be the algebra of
observables associated with $X\subset V$. We consider the dynamics
$\tau_t$ generated by the Hamiltonian
$$
H_V=\sum_{X\subset V} \Phi(X).
$$  The main assumptions are as follows:
The interactions $\Phi(X)$ are uniformly bounded and finite range,
and, for convenience, we will take pair interactions with range 1.
So, for all $X\subset V$, $\Phi(X) = \Phi(X)^*\in \A_X$, $\| \Phi_X
\| \leq J$, for some constant $J>0$, and $\Phi(X) = 0$ if
$\mbox{diam} (X) >1$. Here, for $X\subset V$,
$$
\mbox{diam} (X) = \max \{ d(x,y) \mid x,y\in X\}.
$$
Moreover, the boundary of $X$, denoted by $\partial X$, is
$$
\partial A = \{ x\in A \mid \mbox{there exists } y\in V\setminus A,  \mbox{with } d(x,y) \leq 1\},
$$
and we define the distance between sets $X$ and $Y$ to be $$d(X,Y) = \min \{ d(x,y) \mid x\in X, y\in Y\}.$$
For an observable $A$, we define $\tau_t^\Lambda(A) = e^{i t H_\Lambda} A e^{-it H_\Lambda}$, where $H_\Lambda = \sum_{X \subset \Lambda} \Phi(X)$. 

In order to construct the approximation to the ground state given below, we will assume further that the Hamiltonian $H_V$ has a unique, normalized ground state, which we will denote by $\ket{\Psi_0}$, and a spectral gap $\gamma >0$ to the first excited state.
\begin{theorem} \label{thm:appgs}
There exists $\xi > 0$, such that for any sufficiently large $\ell \ge c_0\, d^2 \xi^2$
there exist two orthogonal projections $P_A \in \A_A$, 
$P_{V\setminus A} \in \A_{V\setminus A}$ and an operator 
$P_B\in \A_{B(A;\ell)}$ with $\|P_B\| \le 1$, such that
\begin{equation} \label{eq:main}
\| P_B P_A P_{V\setminus A} - P_0 \| \leq C(\gamma,d,J) | \partial A |^2 e^{- \ell/ \xi}
\end{equation}
where $P_0$ is the projection onto the (unique) ground state, $c_0$ is of order unity and $C(\gamma,d,J)$, $\xi$ are explicit in terms of $J$, $\gamma$, and the
dimensionality of $V$.
\end{theorem}
The following lemma is an essential tool for the
approximation theorem, so we state it and prove it here.
\begin{lemma}[Lieb-Robinson Bound]\label{LR:bound}
Let $H_V$ be a Hamiltonian satisfying the above assumptions. Then, for a region $\Lambda \subset V$ and two observables $A\in \A_X$ and $B\in \A_Y$, with $X \subset \Lambda$ and $Y \cap \Lambda \neq \emptyset$, we have the following bound:
\be
\|[\tau_t^\Lambda(A),B]\| \le 2\|A\| \|B\| |\partial X| e^{-\mu d(X,Y)} e^{v|t|}, \qquad \forall \, |t| \le e^{-(1+\mu)} \frac{d(X,Y)}{v},
\ee
where $v = 4(2d-1) J$ is the Lieb-Robinson velocity of the system, and $\mu$ is a free parameter.
\end{lemma}
We should note here that similar to the bound given in \cite{loc_est}, the above bound is actually true for all times $t$, since for large $|t|$ and fixed $d(X,Y)$, $\mu$ becomes negative, making the above upper bound non-optimal, as one always has the naive bound $\|[\tau_t^\Lambda(A),B]\| \le 2\|A\| \|B\|$. The above ``constraint'' in $|t|$ is given only to clarify the regime in which the Lieb-Robinson bound can be more useful than the naive upper bound given above.
\begin{proof}
Following closely the proof in \cite{loc_est}, we define the quantity $$C_B(Z;t) = \sup_{A\in \A_Z} \frac{\|[\tau_t^{\Lambda}(A),B]\|}{\|A\|},
$$ 
where $B \in \A_Y$ is fixed, and $Z \subset \Lambda$ is considered arbitrary. For $A \in \A_X$, we introduce the operator function $f(t) = [\tau_t^{\Lambda}(\tau_{-t}^X(A)), B]$. Noting that $\tau_{-t}^X(A)\in \A_X$ and that $\tau_t^{\Lambda}([A,B]) = [\tau_t^{\Lambda}(A),\tau_t^{\Lambda}(B)]$ we may compute
\begin{eqnarray*}
f'(t) &=& i \left[[\tau_t^{\Lambda}\left(H_{S(X)}\right),\tau_t^{\Lambda}(\tau_{-t}^{X}(A))], B\right]\\
&=& i \left[\tau_t^{\Lambda}\left(H_{S(X)}\right), f(t)\right] - i \left[\tau_t^{\Lambda}(\tau_{-t}^{X}(A)),\left[\tau_t^{\Lambda}\left(H_{S(X)}\right), B\right] \right],
\end{eqnarray*}
where we used the Jacobi identity $[A,[B,C]]+[B,[C,A]]+[C,[A,B]]=0$ for the second equality, and defined $S(X)$ to be the set of nearest-neighbor interactions $\{Z\subset \Lambda\setminus X: Z\cap \partial X \ne \emptyset \}$.
Now, one may easily verify that $f(t)$ is given by:
\benn
e^{i\A(t)}\left(f(0) -i \int_0^t e^{i\A(-s)}\left(\left[\tau_s^{\Lambda}(\tau_{-s}^{X}(A)),\left[\tau_s^{\Lambda}\left(H_{S(X)}\right), B\right] \right]\right) e^{-i\A(-s)}\, ds \right) e^{-i\A(t)},
\eenn
where $\A(t) = \int_0^t \tau_s^{\Lambda}(H_{S(X)})\, ds$. Since $\A(t)$ is self-adjoint, we
have that:
\begin{eqnarray}
\|f(t)\| &\le& \|f(0)\| + \int_0^t \| [\tau_s^{\Lambda}(\tau_{-s}^{X}(A)),[\tau_s^{\Lambda}(H_{S(X)}), B] ]\| ds\\
&\le& \|[A,B]\| + 2\|A\|  \sum_{Z\in S(X)} \int_0^t \|[\tau_s^{\Lambda}(\Phi(Z)), B]]\| ds\label{ft:bound}
\end{eqnarray}
From (\ref{ft:bound}) and the fact that $\|\tau_{-s}^{X}(A)\| = \|A\|$, we get the following recursive relation:
\be\label{cb:bound}
C_B(X;t) \le C_B(X;0) + 2\sum_{Z\in S(X)} \|\Phi(Z)\| \int_0^t C_B(Z,s)\,ds
\ee
Moreover, it should be clear that $C_B(Z,0) \le 2\|B\| \delta_Y(Z)$, where $\delta_Y(Z) = 0$ if $Y\cap Z = \emptyset$ and $\delta_Y(Z)=1$, otherwise.
Iterating (\ref{cb:bound}) we get:
\be\label{cb:bound2}
C_B(X;t) \le 2\|B\| \sum_{n=0}^{\infty} \frac{(2|t|)^n}{n!} a_n,
\ee
where $$a_n = \sum_{Z_1\in S(X)} \sum_{Z_2\in S(Z_1)} \cdots \sum_{Z_{n}\in S(Z_{n-1})} \delta_Y(Z_n) \Pi_{i=1}^n \|\Phi(Z_i)\|,$$ for $n\ge 1$. Since $|S(Z)|\le 2(2d-1)$ for $Z$ with $\diam(Z)\le 1$, we have the following upper bound on $a_n$:
\be\label{a_n:bound}
a_n \le \left(\frac{|S(X)|}{4d-2}\right) \left[2(2d-1)J\right]^{n} \delta_Y(S^{(n)}(X)),
\ee
where $S^{(n)}(X) = S(S(\cdots S(S(X))\cdots))$, composition $n$ times. Now, note that the set 
$S^{(n)}(X)$ contains sites at most distance $n$ away from $X$; that is $d(S^{(n)}(X), X) = n$. Hence,
$d(S^{(n)}(X), Y) \ge d(X,Y)-n$, which implies that $a_n = 0$, for $n < d(X,Y)$.
Going back to (\ref{cb:bound2}) and noting that $\left(\frac{|S(X)|}{4d-2}\right) \le |\partial X|$ we get:
\begin{eqnarray}
C_B(X;t) &\le& 2\|B\| |\partial X|  \sum_{n=d(X,Y)}^{\infty} \frac{(2|t|)^n}{n!} a_n\nonumber\\
&\le& 2\|B\| |\partial X| \, \frac{(4(2d-1)J|t|)^{d(X,Y)}}{d(X,Y)!} \sum_{n=0}^{\infty} \frac{(4(2d-1)J|t|)^n}{n!}\nonumber\\
&\le& 2\|B\| |\partial X| \, \frac{(v|t|)^{d(X,Y)}}{d(X,Y)!} e^{v|t|}\label{cb:bound3},
\end{eqnarray}
where $v =4(2d-1)J$. Now, using the simple bound $\ln n! \ge n \ln (n/e)$ for $n\ge 1$, which one may easily prove by induction, we have that:
\be
\frac{(v|t|)^{d(X,Y)}}{d(X,Y)!} \le e^{-d(X,Y) \left(\ln \left[\frac{d(X,Y)}{e\,v|t|}\right]\right)} \le e^{-\mu d(X,Y)}, \qquad |t| \le e^{-(1+\mu)} \frac{d(X,Y)}{v},
\ee
which combined with (\ref{cb:bound3}) and the definition of $C_B(X;t)$ completes the proof.
\end{proof}

We are now ready to prove Theorem \ref{thm:appgs}. The proof consists of three main steps, which
translate into Propositions \ref{prop:bd}, \ref{prop:bd2} and \ref{prop:bd3}, respectively. Here is a brief overview of the proof:

The first step is to use the spectral gap in order to transform the original Hamiltonian into one consisting of $3$ parts, each of which has exponentially small energy in the ground state with respect to some length $\ell$. In particular, the $3$ pieces are chosen so that the middle piece, which is a boundary-like region of increasing thickness $\ell$, acts as an interaction bridge between the other two pieces of the Hamiltonian. 

Nevertheless, at this point, the transformation of the Hamiltonian will introduce long-range interactions within each of the $3$ components, which becomes an obstacle when one wishes to isolate the contribution of each piece to the ground state of the original Hamiltonian. In order to reclaim a strictly local range of interactions for each component, we use the Lieb-Robinson bound given in Lemma \ref{LR:bound} to concentrate the interactions within each component in regions of larger, but strictly local, support. 

At this stage, we will have transformed the original Hamiltonian into a new one that is close in norm to the original and has, furthermore, the desirable property of consisting of components each of which has small energy in the ground state and whose support is localized within a region of our choice. It remains now to construct projectors onto the low energy states of each component and combine them to produce the desired approximation to the ground state.

After this brief overview, we are now ready to prove the theorem.

\begin{proof}
First we note that  $H_V$ can be partitioned as a sum
\begin{equation}
H_V = H_I + H_B + H_E,
\end{equation}
where
\begin{equation}
H_I = \sum_{\stackrel{X \subset V:}{X \cap I \neq \emptyset}}
\Phi(X),
\end{equation}
\begin{equation}
H_B = \sum_{\stackrel{X \subset V:}{X \subset B}} \Phi(X),
\end{equation}
and
\begin{equation}
H_E = \sum_{\stackrel{X \subset V:}{X \cap E \neq \emptyset}}
\Phi(X).
\end{equation}
Moreover, we have that
\begin{lemma} \label{lem:commutator}
The following commutators have norms bounded by the number of
interactions at the boundary of the relevant regions:
\be
\| [H_V,H_I] \| \leq 8d^2 J^2 |\partial I|, \quad \| [H_V,H_E] \| \leq 8d^2 J^2 |\partial E|,
\ee 
and
\be
\| [H_V,H_B] \| \leq
8d^2 J^2 \left( |\partial I| + | \partial E| \right).
\ee
\end{lemma}
We prove this lemma after the proof of Theorem
\ref{thm:appgs}.

Without loss of generality, we may add a constant
to each of the above Hamiltonians  and thus assume that their ground
state energy is zero. In other words, we really work with the
Hamiltonians $\tilde{H}_X = H_X - \langle \Psi_0, H_X \Psi_0
\rangle$, for  $X \in \{V, I, B, E \}$.

In this case, each term has zero ground state expectation. 
However, it is not clear that these terms have a quantifiably small norm 
when applied to the ground state. To achieve small norms, we introduce
non-local versions of these observables.
Given any self-adjoint Hamiltonian $H$, any local observable $O$,
and any $\alpha >0$, we may define the (non-local) observable
\begin{equation}
(O)_{\alpha} = \sqrt{ \frac{ \alpha}{ \pi}} \int_{- \infty}^{\infty}
\tau_t(O) e^{- \alpha t^2} dt,
\end{equation}
where here $\tau_t(O) = e^{itH}Oe^{-itH}$. It is easy to see that
for every $\alpha >0$, $\| (O)_{\alpha} \| \leq \| O\|$.

With respect to the Hamiltonian $H_V$ and $\alpha >0$ (a free
parameter to be chosen later), we now consider the non-local
operator
\begin{equation}
( H_X )_{\alpha} \, =  \, \sqrt{ \frac{ \alpha}{ \pi}} \int_{-
\infty}^{\infty} \tau_t ( H_X ) \, e^{- \alpha t^2} \, dt,
\end{equation}
for $X\in\{I,B,E\}$.

Note that these non-local observables still sum to the total
Hamiltonian, i.e.,
\begin{eqnarray}\label{eq:nonlocH}
( H_I)_{\alpha} \, + \, ( H_B)_{\alpha} \, + \, ( H_E)_{\alpha} =
H_V.
\end{eqnarray}
Also, they have zero ground state expectation, i.e.,
\begin{equation}
\langle \psi_0, ( H_I)_{\alpha} \psi_0 \rangle \, = \, \langle
\psi_0, ( H_B)_{\alpha} \psi_0 \rangle \, = \, \langle \psi_0, ( H_E
)_{\alpha} \psi_0 \rangle \, = \, 0.
\end{equation}
Lastly, each of these observables, when applied to the ground state,
has a small vector norm.
\begin{proposition} \label{prop:bd} Under the assumptions above, we have that for any $\alpha >0$,
\begin{equation}
\| ( H_X)_{\alpha} \psi_0 \| \, \leq \, \frac{ \| [H_V, H_X ] \| }{
\gamma} e^{- \frac{\gamma^2}{4 \alpha}},
\end{equation}
for $X\in\{I, B, E\}$.
\end{proposition}

We will prove Proposition \ref{prop:bd} after we finish
the proof of this theorem. We want the above bound to have an exponential decay 
of the form $e^{-\ell/\xi'}$ for some positive $\xi'$ that will be determined later. Thus, we set
\be\label{alpha}
\alpha = \frac{\gamma^2 \xi'}{4 \ell}.
\ee
This completes the first step in the proof of Theorem \ref{thm:appgs}.
The next step is to
approximate each of $( H_I)_{\alpha}$, $( H_B)_{\alpha}$, and $(H_E)_{\alpha}$ 
with observables $M_I(\alpha), M_B(\alpha)$ and
$M_E(\alpha)$. These observables will be constructed in such a way
that each has support on a specific subset of $V$ and a small vector norm when
applied to the ground state.

For $X \in \{I,B,E\}$, set
\begin{equation} \label{eq:mlax}
M_X(\alpha) \, = \,  \sqrt{ \frac{ \alpha}{ \pi}} \int_{-
\infty}^{\infty} \sigma_t^{X} \left( H_X \right) e^{- \alpha t^2} dt.
\end{equation}
Here, for any local observable $O$, we have introduced the following
evolutions
\begin{eqnarray} \label{eq:locham}
\sigma_t^{I} = \tau_t^{A}(O) &=& e^{itH_A} O e^{-itH_A}\\
\sigma_t^{B} = \tau_t^{B(A;2\ell)}(O) &=& e^{itH_{B(A;2\ell)}} O e^{-itH_{B(A;2\ell)}} \\
\sigma_t^{E} = \tau_t^{V\setminus A}(O) &=& e^{itH_{V \setminus A}} O e^{-itH_{V \setminus
A}}.
\end{eqnarray}

Observe that the local Hamiltonians defined above have been chosen
such that $\mbox{supp}(M_I(\alpha)) = A$, $\mbox{supp}(M_B(\alpha))
= B(A;2\ell)$, and $\mbox{supp}(M_E(\alpha)) = V\setminus A$.

We now state and prove an upper bound on the number of lattice points on the surface of a sphere in $\Ir^d$.

\begin{lemma}\label{s_n:bound}
Let $s(n,d)$ denote the number of points on a sphere of radius $n$ in $\Ir^d$.
Then, the following bound holds:
\be
s(n,d) \le 2^d n^{d-1}.
\ee
\begin{proof}
Spheres in $\Ir^d$ look like $d-$dimensional diamonds.
The sphere of radius $n$ may be constructed by choosing any of the $d$ hyperplanes going through an ``equator'' of a sphere of radius $n-1$, and lifting the two ``hemispheres'' in a direction normal to the chosen hyperplane to insert an equator of radius $n$ between them. Noting that an equator of radius $n$ of a $d$-dimensional sphere is just a $(d-1)$-dimensional sphere of radius $n$, we get the following recursive relation:
\be\label{s_n}
s(n,d) = s(n-1,d) + s(n,d-1) + s(n-1,d-1), \qquad s(1,d) = 2d, \, s(n,1) = 2,
\ee
where the term $s(n-1,d-1)$ comes from adding the missing equator to one of the
two lifted hemispheres that make up the new polar caps.
Note that the above bound is true for $d=1,2$, with equality. Moreover, we have from (\ref{s_n}) that
$$
s(n,d) \le s(n-1,d) + 2s(n,d-1) \implies s(n,d) \le 2 \sum_{k=1}^n s(k,d-1) +s(1,d) - 2s(1,d-1),
$$
and noting that $s(1,d) - 2s(1,d-1) = -2(d-2) \le 0,$ for $d\ge 2$, we get
$$
s(n,d) \le 2 \sum_{k=1}^n s(k,d-1) \le 2\,n\,s(n,d-1) \le (2n)^{d-1} \, s(n,1), \qquad d\ge2.
$$
which proves (\ref{s_n:bound}). 
\end{proof}
\end{lemma}
We will now use the finite (Lieb-Robinson) velocity $v = 4(2d-1)J$ of our dynamics, to show that
our newly defined local Hamiltonian operators are close in norm to the original ones.

\begin{proposition} \label{prop:bd2} Set $\xi' = 4+4\left(\frac{v}{\gamma}\right)^2$. 
For $X \in \{ I, E\}$, the following bound holds 
\begin{eqnarray} \label{eq:difbd}
&&\left\| (H_X)_{\alpha} \, - \, M_X(\alpha) \right\| \le \nonumber\\
&& e^{-\ell/\xi'} \left\{ \left(\frac{4}{\sqrt{\pi} \sqrt{\gamma^2+v^2}}\ell^{1/2}
|\partial X|+ 3^{d+2}\, 4^d \,d^{d-1}\ell^{d-1} | \partial A |\right)
\frac{4d^2J^2}{v} \right\}
\end{eqnarray}
and similarly,
\begin{eqnarray*} \label{eq:difbd2}
&&\left\| ( H_B)_{\alpha} \, - \, M_B(\alpha) \right\| \le \nonumber\\
&& e^{-\ell/\xi'} 
\left\{ \left(\frac{4}{\sqrt{\pi} \sqrt{\gamma^2+v^2}} {\ell}^{1/2}
(|\partial B(A;\ell)|)+ 3^{d+2}\, 4^d \,d^{d-1}\ell^{d-1} | \partial B(A; 2\ell) |\right) 
\frac{4d^2J^2}{v} \right\}.
\end{eqnarray*}
\end{proposition}
The proof of Proposition \ref{prop:bd2} appears after the proof of Proposition \ref{prop:bd} below. 
Now, for any $X \in \{ I, B, E \}$, the bound
\begin{equation} \label{eq:mxpsi}
\| M_X( \alpha) \psi_0 \| \, \leq \, \| (H_X)_{\alpha} \psi_0 \| \,
+ \, \| \left( M_X(\alpha) - ( H_X)_{\alpha} \right) \psi_0 \|
\end{equation}
is trivial. Thus, given Proposition \ref{prop:bd} and
\ref{prop:bd2}, when applied to the ground state, each of these
local observables will have small vector norm along the perscribed
parametrization.

As an immediate consequence of Proposition \ref{prop:bd2} and Lemma \ref{s_n:bound}, we have the following corollary.

\begin{corollary}\label{cor:bd2}
Set $D = \max \{\frac{1}{2}, 2(d-1)\}$ and 
$$C_1(\gamma, d,J) = \left(\frac{2^{d+2}}{\sqrt{\pi} \sqrt{\gamma^2+ v^2}} + \frac{2v}{\gamma}+3^{d+2} \,4^{2d} \, d^{d-1}\right)\frac{(2dJ)^2}{v}.$$ Then, the bounds
\begin{equation}
\left\| H_V \, - \, \left(M_I( \alpha) + M_B(\alpha) +M_E( \alpha)
\right) \right\| \, \leq 2\,C_1(\gamma, d,J)\, |\partial A |\, \ell^{D} e^{- \frac{ \ell}{\xi'}}
\end{equation}
and
\begin{equation}
\max_{X \in \{I,B,E \}} \left\| M_X( \alpha) \Psi_0 \right\| \, \leq
\, C_1(\gamma, d,J)\, | \partial A |\, \ell^{D} e^{- \frac{ \ell}{\xi'}}
\end{equation}
are valid.
\end{corollary}
\begin{proof}
{F}rom the definitions of the relevant sets, it is clear that since for each point in $\partial A$ there
are $s(\ell,d)$ sites at a distance $\ell \ge 1$, using Lemma \ref{s_n:bound} we have
\be
\max\{2|\partial A|, | \partial I| +| \partial E|, | \partial B(A;\ell)|\} \le  \, 2^{d} \ell^{d-1} | \partial A |.
\ee
and similarly,
\be
| \partial B(A;2\ell)| \le  \, 2^{2d-1} \ell^{d-1} | \partial A |.
\ee
Now, the first bound follows from \eqref{eq:nonlocH} and
Proposition \ref{prop:bd2}.
Similarly, the second bound is obtained by combining
\eqref{eq:mxpsi} with Propositions \ref{prop:bd} and
\ref{prop:bd2}. Note that the constant $C_1(\gamma, d,J)$
was chosen optimally for the second bound, but one 
could eliminate the $\frac{2v}{\gamma}$ term for the first bound.
\end{proof}

The next step in the proof of this theorem is to define an
approximate ground state projector. Consider
\begin{equation}
\tilde{P}_{\alpha} \, = \, \sqrt{ \frac{\alpha}{\pi}} \int_{-
\infty}^{\infty} e^{i H_V t} e^{- \alpha t^2} \, dt.
\end{equation}
Observe that for any functions $f$ and $g$,
\begin{eqnarray}
\langle f, ( \tilde{P}_{\alpha} - P_0) g \rangle & = &  \sqrt{
\frac{\alpha}{\pi}} \int_{- \infty}^{\infty} e^{- \alpha t^2} \,
\int_0^{\infty} e^{i E t} d \langle f, P_E g \rangle dt \, - \,
\langle f, P_0 g \rangle \nonumber \\ & = & \int_{\gamma}^{\infty}
\sqrt{ \frac{\alpha}{\pi}} \int_{- \infty}^{\infty} e^{i E t} e^{-
\alpha t^2} \, dt \, d \langle f, P_E g \rangle
 \nonumber \\ & = & \int_{\gamma}^{\infty}  e^{- \frac{E^2}{4 \alpha}} \, d \langle f, P_E g \rangle,
\end{eqnarray}
which then readily yields that
\begin{equation} \label{eq:approxbd}
\left\| \tilde{P}_{\alpha} - P_0 \right\| \, \leq \, e^{- \frac{
\gamma^2}{4 \alpha}} \, = \, e^{- \frac{\ell}{\xi'}}
\end{equation}

We now define a first approximation of this ground state projector.
Set
\begin{equation}
\hat{P}_{\alpha} \, = \,  \sqrt{ \frac{ \alpha}{ \pi}}
\int_{-\infty}^{\infty} e^{i(M_I+M_B+M_E)t} e^{- \alpha t^2} \, dt,
\end{equation}
where we have dropped the dependence of $M_I$, $M_B$, and $M_E$ on
$\alpha$. Clearly,
\begin{equation} \label{eq:hatbd}
\left\| \hat{P}_{\alpha} - P_0 \right\| \, \leq  \, \left\|
\hat{P}_{\alpha} - \tilde{P}_{\alpha} \right\| \, + \, \left\|
\tilde{P}_{\alpha} - P_0 \right\|.
\end{equation}
The final term above we have bounded in (\ref{eq:approxbd}). To see
a bound on the first term, we introduce the function
\begin{equation}
F_t( \lambda) \, = \, e^{i \lambda (M_I+M_B+M_E)t} \, e^{i(1-
\lambda) H_V t}.
\end{equation}
One easily calculates that
\begin{equation}
F_t'( \lambda) \, = \, - \, i \, t \,  e^{i \lambda (M_I+M_B+M_E)t}
\left\{ H_V \, - \, (M_I+M_B+M_E) \right\} e^{i(1- \lambda)H_V t}.
\end{equation}
Thus clearly,
\begin{eqnarray}
\left\| \hat{P}_{\alpha} - \tilde{P}_{\alpha} \right\| & \leq &  \sqrt{ \frac{ \alpha}{ \pi}} \int_{-\infty}^{\infty} \left\| F_t(1) - F_t(0) \right\| \, e^{- \alpha t^2} \, dt \nonumber \\
& \leq & \sqrt{ \frac{ \alpha}{ \pi}} \int_{-\infty}^{\infty} \left\|  H_N \, - \, (M_I+M_B+M_E) \right\| \, |t| \, e^{- \alpha t^2} \, dt \nonumber \\
& = & \frac{1}{\sqrt{\pi \alpha}} \left\|  H_N \, - \, (M_I+M_B+M_E) \right\|\\
&\le & \frac{2\,C_1(\gamma, d,J)}{\sqrt{\pi} \sqrt{\gamma^2+v^2}} \, |\partial A |\, \ell^{D+1/2} e^{- \frac{ \ell}{\xi'}},
\end{eqnarray}
along the given parametrization for $\alpha$ and $\xi'$. Thus, \eqref{eq:hatbd} now yields
\begin{equation} \label{eq:hatexactbd}
\left\| \hat{P}_{\alpha} - P_0 \right\| \, \leq  \,
\left\{\frac{2\,C_1(\gamma, d,J)}{\sqrt{\pi} \sqrt{\gamma^2+v^2}} \, |\partial A |\, \ell^{D+1/2}+1\right\}e^{-
\ell/\xi'}.
\end{equation}

We now define the spectral projections $\PA$ and $\PV$
which appear in the statement of Theorem \ref{thm:appgs}.

Denote by 
$$ c= C_1(\gamma, d,J)\, |\partial A |\, \ell^{D+1/2} e^{- \frac{\ell}{2\xi'}}.$$
 For $X
\in\{I,E\}$, define $\PA$ and $\PV$ to be the spectral projection,
corresponding to the self-adjoint matrices $M_I$ and $M_E$, respectively, onto those eigenvalues
less than $c$. In light of Corollary \ref{cor:bd2}, the following bounds:
\begin{equation} \label{eq:oxbd}
\left\| \left( 1 - \PA \right) \Psi_0 \right\| \, \leq \,
\frac{1}{c} \left\| M_I \Psi_0 \right\| \, \leq \, e^{- \ell/ 2
\xi'},
\end{equation}
easily follow, and the same bound holds for $1-\PV$. 
To see that we can always find a non-zero $\PA$ (and similarly $\PV$), we simply observe that
the above inequality for $\PA$ can be written as
\begin{equation}
\langle\Psi_0| \PA|\Psi_0\rangle \, \geq \, 1-e^{- \ell/ 2\xi'}
\end{equation}
which implies that $\PA$ is non-zero.
Moreover, it is obvious that both $\PA$ and $\PV$ satisfy the required support 
assumption.

These spectral projections may be inserted into our previous
estimates as
\begin{equation}
\left\| \hat{P}_{\alpha} \PA \PV - P_0 \right\| \, \leq \,  \left\|
\left( \hat{P}_{\alpha} \, - \, P_0 \right) \PA \PV \right\| \, + \,
  \left\| P_0 \left( 1 \, - \, \PA \PV \right) \right\|,
\end{equation}
here we have dropped the dependence on $A$ and $\ell$. The first
term above is estimated as in (\ref{eq:hatexactbd}) above. Since
\begin{equation}
(1-\PA\PV) \, = \, \frac{1}{2} \left\{ (1-\PA)(1+\PV) \, + \,
(1-\PV)(1+\PA) \right\},
\end{equation}
it is clear that
\begin{equation}
 \left\| P_0 \left( 1 \, - \, \PA \PV \right) \right\| \, \leq \, \| P_0 (1 - \PA) \| \, + \, \| P_0 (1- \PV) \| \, \leq \, 2  e^{- \ell/ 2 \xi'},
\end{equation}
using (\ref{eq:oxbd}). Therefore, we now have that
\begin{equation}
\left\| \hat{P}_{\alpha} \PA \PV - P_0 \right\| \, \leq
\,\left(\frac{2\,C_1(\gamma, d,J)}{\sqrt{\pi} \sqrt{\gamma^2+v^2}} \, |\partial A |\, \ell^{D+1/2}+3\right)e^{-
\ell/2\xi'}.
\end{equation}

We are now left with the task of finding a local operator $\PB \in \A_{B(A;3\ell)}$ that approximates $\hat{P}_{\alpha}$. In order to do so let us
write
\begin{equation}
\hat{P}_{\alpha} \PA \PV \, = \,  \sqrt{ \frac{ \alpha}{ \pi}}
\int_{-\infty}^{\infty} e^{i(M_I+M_B+M_E)t} e^{-i(M_I+M_E)t}
e^{i(M_I+M_E)t} \PA\PV e^{- \alpha t^2} \, dt.
\end{equation}
Since the supports of $M_I$ and $M_E$ are disjoint, it is clear that
\begin{eqnarray*}
e^{i(M_I+M_E)t} \PA\PV \, - \, \PA \PV \, &=& \frac{1}{2} 
\left(e^{iM_It} \PA-\PA\right)\left(e^{iM_Et}\PV+\PV\right) \\
&+& \frac{1}{2}\left(e^{iM_Et}\PV-\PV\right)\left(e^{iM_It}\PA+\PA\right).
\end{eqnarray*}
Moreover, we have that
\begin{eqnarray}
\left| \langle f, (e^{iM_It}-1)\PA g \rangle \right| & = & \left| \int_0^{c} (e^{iEt} - 1) \,  d \langle f, P_E^I \PA g \rangle \right| \nonumber \\
& \leq & c \, |t| \, \| f \| \, \| g \|.
\end{eqnarray}
and similarly for $\left| \langle f, (e^{iM_Et}-1)\PV g \rangle \right|$.
 If we now define the following operator,
\begin{equation}
P_B( \alpha) \, = \, \sqrt{ \frac{ \alpha}{ \pi}}
\int_{-\infty}^{\infty} e^{i(M_I+M_B+M_E)t} e^{-i(M_I+M_E)t} e^{-
\alpha t^2} \, dt,
\end{equation}
then we have just demonstrated that
\begin{eqnarray}
&&\| P_B(\alpha) \PA \PV - P_0 \| \nonumber\\
&& \leq \| P_B(\alpha)  \PA \PV - \hat{P}_{\alpha}\PA\PV \| \, + \, \| \hat{P}_{\alpha}\PA\PV - P_0 \| \nonumber \\
&& \leq  2 \sqrt{ \frac{ \alpha}{ \pi}} \int_{-\infty}^{\infty} c
|t|  e^{- \alpha t^2} \, dt \, + \,\left(\frac{2\,C_1(\gamma, d,J)}{\sqrt{\pi} \sqrt{\gamma^2+v^2}} \, |\partial A |\, \ell^{D+1/2}+3\right)e^{-\ell/2\xi'} \nonumber \\
&&\leq \left(\frac{4\,C_1(\gamma, d,J)}{\sqrt{\pi} \sqrt{\gamma^2+v^2}} \, |\partial A |\, \ell^{D+1/2}+3\right)e^{- \ell/2\xi'} \nonumber\\
&&\le C_2(\gamma,d,J) |\partial A |\, \ell^{2d-1} e^{- \ell/2\xi'} ,
\end{eqnarray}
where we have used the simple bound $D \le 2d-3/2$ and set $$C_2(\gamma,d,J)=\frac{4\,C_1(\gamma, d,J)}{\sqrt{\pi} \sqrt{\gamma^2+v^2}}+3$$ in the last line.

Given our arguments above, we need only show that $P_B( \alpha)$ may
be approximated by a local operator $\PB$.
\begin{proposition} \label{prop:bd3}
There exists a local operator $\PB$ with support in $B(A; 3 \ell)$ and $\|\PB\| \le 1$,
for which
\begin{equation} \label{eq:locpbd}
\| \PB - P_B( \alpha) \| \, \leq \, C_3(\gamma,d,J) | \partial A|^2 \ell^{3d} e^{-\frac{\ell}{\xi'}}.
\end{equation}
where $$C_3(\gamma, d, J) = \frac{4^dJ}{v} \left(\frac{2v}{\sqrt{\pi}\sqrt{\gamma^2+v^2}} + 4^{d} + \frac{2v}{\gamma^2+v^2} C_1(\gamma,d,J)\right)$$ and $C_1(\gamma,d,J)$ is given in Corollary \ref{cor:bd2}.
\end{proposition}

The proof of Proposition \ref{prop:bd3} can be found at the end of this section. It is
clear that once we have proven Proposition \ref{prop:bd3}, the
estimate (\ref{eq:main}) readily follows from the triangle inequality
\begin{eqnarray}
\|\PB\PA\PV - P_0\| &\le& \|\PB(\alpha)\PA\PV - P_0\| + \|\PB-\PB(\alpha)\| \nonumber\\
&\le& (C_2(\gamma,d,J)+C_3(\gamma,d,J)) | \partial A|^2 \ell^{3d} e^{-\frac{\ell}{2\xi'}}
\end{eqnarray}
after setting $\ell'=3\ell$, since $\PB \in B(A;\ell')$ in the statement of theorem, and setting $\xi = 25/4 \xi' = 25 (1+(\frac{v}{\gamma})^2)$, since the final decay above is of the form $e^{-\frac{\ell'}{6\xi'}}$ and we wish to cover the factor $\ell^{3d}$ by choosing a slower exponential decay for $\ell'$ large enough.
 
This completes the proof of Theorem \ref{thm:appgs}.

\end{proof}
Now we prove the various Propositions used in the proof of the
Theorem.
\begin{proof}[Proof of Lemma \ref{lem:commutator}]
Define
\begin{equation}
C_I = \left[H_V, H_I \right], \quad C_B = \left[H_V, H_B \right],
\quad \mbox{and} \quad C_E = \left[H_V, H_E \right].
\end{equation}
{F}rom the definitions of the sets $I$, $B$, and $E$, it is easy to
see that $C_I = \left[H_B, H_I \right]$, $C_E = \left[H_B, H_E
\right]$, and $C_B = \left[H_I, H_B \right] + \left[H_E, H_B
\right]$. Clearly, for $Z \in \{ I, E \}$, we have that
\begin{eqnarray*}
\left[ H_Z, H_B \right] & = & \sum_{\stackrel{X \subset V:}{X \cap Z
\neq \emptyset}} \left[ \Phi(X), H_B \right]
 \nonumber \\
&=& \sum_{z \in \partial Z} \sum_{i=1}^d \chi_{Z^c}(z+e_i) \left[ \Phi( \{z, z+e_i \}), H_B \right] \nonumber \\
&+& \sum_{z \in \partial Z} \sum_{i=1}^d \chi_{Z^c}(z-e_i) \left[ \Phi( \{z, z-e_i \}), H_B \right] \nonumber \\
&=& \sum_{z \in \partial Z} \sum_{i=1}^d \chi_{Z^c}(z+e_i) \sum_{\stackrel{w \in B:}{d(z+e_i,w) =1}} \left[ \Phi( \{z, z+e_i \}), \Phi( \{z+e_i, w \}) \right] \nonumber \\
&+& \sum_{z \in \partial Z} \sum_{i=1}^d
\chi_{Z^c}(z-e_i) \sum_{\stackrel{w \in B:}{d(z-e_i,w)=1}} \left[
\Phi( \{z, z-e_i \}), \Phi( \{z-e_i, w \}) \right],
\end{eqnarray*}
from which we get:
\begin{equation}
\left\| \left[ H_Z, H_B \right] \right\| \leq 2\left(|\partial Z | \, d\, (2d-1)\, (2J^2)\right).
\end{equation}
The bound
\begin{equation*}
\left\| \left[ H_Z, H_B \right] \right\| \leq 8d^2 J^2|\partial Z |,
\end{equation*}
is now clear.
\end{proof}

Here is the proof of Proposition \ref{prop:bd}.

\begin{proof}[Proof of Proposition \ref{prop:bd}]
We begin by observing that
\begin{eqnarray}
\| ( H_X )_{\alpha} \Psi_0 \| & \leq & \frac{1}{\gamma} \| H_V ( H_X)_{\alpha} \Psi_0 \| \nonumber \\
& = & \frac{1}{\gamma} \| [ H_V, (H_X)_{\alpha}] \Psi_0 \| \nonumber \\
& = & \frac{1}{\gamma} \| (C_X)_{\alpha} \Psi_0 \|.
\end{eqnarray}
For the first inequality above, we used that $(H_X)_{\alpha} \Psi_0$
projects off the ground state, since $\bra{\Psi_0} (H_X)_{\alpha} \ket{\Psi_0} = 0$.

The value of $\| (C_X)_{\alpha} \Psi_0 \|$ can be estimated using
the spectral theorem. For any vector $f$, one has that
\begin{eqnarray}
\langle f, (C_X)_{\alpha} \Psi_0 \rangle & = & \sqrt{ \frac{\alpha}{\pi}} \int_{- \infty}^{\infty} e^{- \alpha t^2} \langle f, \tau_t(C_X) \Psi_0 \rangle dt \nonumber \\
& = &  \sqrt{ \frac{\alpha}{\pi}} \int_{- \infty}^{\infty} e^{- \alpha t^2} \langle f, e^{it H_V}C_X \Psi_0 \rangle dt \nonumber \\
& = &  \sqrt{ \frac{\alpha}{\pi}} \int_{- \infty}^{\infty} e^{- \alpha t^2} \int_{\gamma}^{\infty} e^{itE} d \langle f, P_E C_X \Psi_0 \rangle dt \nonumber \\
& = & \int_{\gamma}^{\infty} \sqrt{ \frac{ \alpha}{ \pi}} \int_{- \infty}^{\infty} e^{- \alpha t^2} e^{itE} \, dt \, d \langle f, P_E C_X \Psi_0 \rangle \nonumber \\
& = & \int_{\gamma}^{\infty} e^{- \frac{E^2}{4 \alpha}} \, d \langle
f, P_E C_X \Psi_0 \rangle.
\end{eqnarray}
In the third equality above we have introduced the notation $P_E$
for the spectral projection corresponding to the self-adjoint
operator $H_V$, and we also used that $\langle \Psi_0, C_X \Psi_0
\rangle = 0$. The last equality is a basic result concerning Fourier
transforms of gaussians (and re-scaling), (see e.g.(5.59) in
\cite{nasi}). We conclude then that
\begin{equation}
\left| \langle f, (C_X)_{\alpha} \psi_0 \rangle \right| \, \leq \,
e^{ - \frac{ \gamma^2}{4 \alpha}} \| f \| \, \| C_X \|,
\end{equation}
and with $f = (C_X)_{\alpha} \psi_0$ we find that
\begin{equation}
\| (C_X)_{\alpha} \psi_0 \| \, \leq \, e^{- \frac{ \gamma^2}{4
\alpha}} \| C_X \|.
\end{equation}
Putting everything together,  we have shown that
\begin{equation} \label{eq:thlbd}
\| ( H_X)_{\alpha} \psi_0 \| \leq \frac{1}{ \gamma} \|
(C_X)_{\alpha} \psi_0 \| \leq \frac{e^{- \frac{ \gamma^2}{4
\alpha}}}{\gamma} \| C_X \| = \frac{\| [H_V,H_X] \|}{\gamma} e^{- \frac{ \gamma^2}{4
\alpha}}.
\end{equation}
\end{proof}

Here is the proof of Proposition \ref{prop:bd2}.

\begin{proof}[Proof of Proposition \ref{prop:bd2}]
Since our method of bounding this difference is similar for $X \in
\{ I, B, E \}$, we will provide the details only in the case of
$X=I$. We begin by observing that
\begin{equation}
\left\| ( H_I)_{\alpha} \, - \, M_I(\alpha) \right\| \, \leq \,
\sqrt{ \frac{ \alpha}{ \pi}} \int_{- \infty}^{\infty} \left\|
\tau_t( H_I) \, - \, \tau_t^I(H_I) \right\| \, e^{- \alpha t^2} dt.
\end{equation}
To bound the integral above, we introduce a parameter $T>0$. For
$|t| > T$, we use that
\begin{eqnarray}
\left\| \tau_t \left( H_I   \right)  \, - \, \tau_t^{A} \left( H_I
\right) \right\| \, &\le& \int_0^{|t|}
\left\| \frac{d}{ds}\Big( \tau_{s} (H_I )  \, - \, \tau_{s}^{I}(H_I) \Big)\right\| ds \nonumber \\
& \leq & 2 \|[H_V, H_I]\| |t|.
\end{eqnarray}
{F}rom this, it follows readily that
\begin{equation}
 \sqrt{ \frac{ \alpha}{ \pi}} \int_{|t|>T} \left\| \tau_t ( H_I)  \, - \, \tau_t^{A}( H_I) \right\|
\, e^{- \alpha t^2} dt \, \leq \, \frac{2 \|[H_V,
H_I]\|}{\sqrt{\alpha\pi}} e^{- \alpha T^2}.
\end{equation}
Again, we want the exponential decay to be of the form $e^{-\ell/\xi'}$, which implies that
\be\label{T}
T = \frac{2\ell}{\gamma \xi'}.
\ee

For $|t| \leq T$, the estimate below is an immediate consequence of
Proposition 1 in \cite{nasi}:
\begin{eqnarray} \label{eq:smt}
\left\| \tau_t ( H_I )  \, - \, \tau_t^{A}( H_I) \right\| \, \leq \,
\int_0^{|t|} \left\| \left[ H_V-H_A, \tau_s^{A} \left( H_{I} \right)
\right]\right\| \, ds.
\end{eqnarray}
The above commutator may be written as
\begin{align}\label{eq:com2}
\left[ H_V - H_A, \tau_s^{A} \left( H_{I} \right) \right] \, = \,&
 \sum_{x \in \partial A} \sum_{i=1}^d \chi_{A^c}(x+e_i)  \left[ \Phi( \{x, x+e_i \}), \tau_s^{A} \left( H_{I}\right)
 \right]\nonumber
 \\&+\sum_{x \in \partial A}\sum_{i=1}^d \chi_{A^c}(x-e_i) \left[ \Phi( \{x, x-e_i \}), \tau_s^{A} \left( H_{I}\right)
 \right],
\end{align}
and for each $x \in \partial A$,
\begin{equation}\label{eq:smt2}
\left[ \Phi( \{x, x \pm e_i \}), \tau_s^{A} \left( H_{I}\right)
\right] =\sum_{\stackrel{Y \subset V:}{Y\cap I \neq \emptyset}}
\left[ \Phi( \{x, x \pm e_i \}), \tau_s^{A} \left( \Phi(Y) \right)
\right].
\end{equation}
Each of these commutators we estimate using the Lieb-Robinson
bound from Lemma \ref{LR:bound}, remembering that it is useful for times $|s| \le e^{-(1+\mu)}\, \ell/v$, since from the definition of $I$, we have $d(Y,\partial A) \ge \ell$ for $Y\cap I\ne \emptyset$. We may easily satisfy this requirement by choosing $T \le e^{-(1+\mu)}\, \ell/v$, which combined with (\ref{T}) implies that 
\be
T = \frac{2\ell}{\gamma \xi'} \le e^{-(1+\mu)}\, \ell/v
\ee
and, hence
\be\label{xi:bound}
1/\xi' \le e^{-(1+\mu)}\, \gamma/(2v)
\ee
Remembering that $x \in \partial A$, Lemma \ref{LR:bound} and (\ref{eq:smt2}) imply the following bound:
\begin{eqnarray}
\left\|\left[ \Phi( \{x, x\pm e_i \}), \tau_s^{A} \left(H_{I}\right)\right]\right\|
& \leq & \sum_{ y \in I} \sum_{\stackrel{y' \in I:}{d(y,y')=1}} \left\|\left[ \Phi( \{x, x \pm e_i \}), \tau_s^{A}
\left( \Phi( \{y,y' \}) \right)\right] \right\| \nonumber \\
& \leq & (2d)\, 2|\partial \{y,y'\}| \|\Phi( \{x,x\pm e_i\})\| \|\Phi( \{y,y' \})\|  e^{v|s| } \nonumber \\
&\times& \sum_{y \in I} e^{- \mu (d(x,y)-1)} \nonumber \\
& \leq & 8d J^2  e^{v|s| } e^{- \mu \ell} \sum_{y \in I} e^{- \mu \left(d(x,y)-(\ell+1)\right)} \label{exp:bound},
\end{eqnarray}
For the last inequality, note that the number of sites $y \in I$ that are a distance $d(x,y) = n+\ell+1, \, n\ge 0$ from $x$ is bounded above by $s(n+\ell+1,d)$, the number of lattice points on a sphere of radius $n+\ell+1$ in $\Ir^d$. We now use Lemma \ref{s_n:bound} to bound the sum $\sum_{y \in I} e^{- \mu (d(x,y)-(\ell+1))}$ in (\ref{exp:bound}):
\begin{eqnarray*}
\sum_{y \in I} e^{- \mu (d(x,y)-(\ell+1))} &\le& \sum_{n\ge 0} s(n+\ell+1,d) e^{- \mu n}\\
&\le& 2^{d} \ell^{d-1} \sum_{n\ge 0} (n+2)^{d-1}e^{- \mu n} \\
&\le& 2^{d} \ell^{d-1} e^{2\mu} \frac{\partial^{\, d-1}}{\partial(-\mu)^{\,d-1}} \left(\sum_{n\ge 0} e^{- \mu n}\right),
\end{eqnarray*}
where we used the bound $\ell+n+1 \le \ell (n+2)$, for $\ell \ge 1, \, n\ge 0$, for the second inequality.
But, $\sum_{n\ge 0} e^{- \mu n} = (1-e^{-\mu})^{-1}$ and one can show inductively that
$$
\frac{\partial^{\, d-1}}{\partial(-\mu)^{\,d-1}} (1-e^{-\mu})^{-1} \le (d-1)! \,(1-e^{-\mu})^{-d} \le \frac{(d-1)!}{\mu^d} \,e^{\mu d}
$$
This demonstrates that
\be
\left\|\left[ \Phi( \{x, x\pm e_i \}), \tau_s^{A} \left(H_{I}\right)\right]\right\| \le 8d!\,e^{2\mu} \left(\frac{2\,e^{\mu}}{\mu}\right)^{d} \, J^2\, \ell^{d-1}\, e^{- \mu \ell} e^{v |s|}
\ee
and feeding back into (\ref{eq:smt}) - (\ref{eq:smt2}), we get
\begin{equation}
\left\| \tau_t ( H_I )  \, - \, \tau_t^{A}( H_I) \right\| \, \leq \, 2d |\partial A|
\frac{8d!\, e^{2\mu}}{v} \left(\frac{2\,e^{\mu}}{\mu}\right)^{d} \, J^2\, \ell^{d-1} e^{- \mu \ell}
\left(e^{v |t|} - 1 \right).
\end{equation}
Letting 
$C(d,\mu) = 
16\, e^{2\mu} \left(\frac{2\,e^{\mu}}{\mu}\right)^{d} d^{d+1}
$ we may now estimate
\begin{eqnarray}
&&\sqrt{ \frac{ \alpha}{ \pi}} \int_{-T}^{T} \left\| \tau_t (H_I)  \,
- \, \tau_t^I(H_I) \right\|\, e^{- \alpha t^2} dt \nonumber\\
& \leq & \frac{C(d,\mu)}{v} \, J^2\, \ell^{d-1} | \partial A |  e^{- \mu \ell} \, \sqrt{ \frac{ \alpha}{ \pi}} \int_{-T}^{T} e^{ v |t|} e^{- \alpha t^2} \, dt \nonumber \\
& = & \frac{C(d,\mu)}{v} \, J^2\, \ell^{d-1} | \partial A | e^{- \mu \ell} e^{\frac{v^2}{4 \alpha}} \sqrt{ \frac{ \alpha}{ \pi}} \int_{-T}^{T} e^{- \alpha \left( t - \frac{ v}{2 \alpha} \right)^2} \, dt \nonumber \\
& \leq &  \frac{C(d,\mu)}{v} \, J^2\, \ell^{d-1} | \partial A |  e^{-
\mu \ell} e^{\frac{v^2}{4 \alpha}}. \label{T:smallbd}
\end{eqnarray}
At this point, since we are interested in exponential decay of the form $e^{-\ell /\xi'}$, (\ref{alpha}) gives the following constraint on $\xi'$:
$$
\frac{1}{\xi'} = \frac{\mu}{1+\left(\frac{v}{\gamma}\right)^2}
$$
Combined with (\ref{xi:bound}), the above constraint implies the following bound on $\mu$:
$$
\mu e^{1+\mu} \le \frac{v^2+\gamma^2}{2v\gamma},
$$
which is satisfied naively for $\mu = 1/4$, since the r.h.s is bounded below by $1$. Hence, we finally have that
\be\label{xi}
\xi' = 4\left(1+\left(\frac{v}{\gamma}\right)^2\right)
\ee
Adding our results for $|t| \ge T$ and $|t| \le T$, it is clear that along the given parametrization (\ref{T}),
\be \label{eq:hlmldif}
\left\|  (H_I)_{\alpha} \, - \, M_I(\alpha) \right\| 
\leq \left\{ \frac{2 \|[H_V, H_I]\|}{\sqrt{\alpha\pi}}+\frac{3(12 d)^{d+1}}{v} \, J^2\, \ell^{d-1} | \partial A |  \right\}e^{-\ell / \xi' },
\ee
where we used the bound $e^{\mu}/\mu \le 6$, for $\mu = 1/4$ to bound $C(d,\mu)$ in (\ref{T:smallbd}).
The bound claimed in (\ref{eq:difbd}) now follows from (\ref{alpha}) and Lemma
\ref{lem:commutator}.
\end{proof}

Here is the proof of Proposition \ref{prop:bd3}.

\begin{proof}[Proof of Proposition \ref{prop:bd3}]
Our goal is to show that the operator
\begin{equation}
P_B( \alpha) \, = \, \sqrt{ \frac{ \alpha}{ \pi}}
\int_{-\infty}^{\infty} e^{i(M_I+M_B+M_E)t} e^{-i(M_I+M_E)t} e^{-
\alpha t^2} \, dt
\end{equation}
is well approximated by a local observable; one with support in 
$B(A;3\ell)$. To do so, we will take the normalized partial trace over the complimentary
Hilbert space, i.e., the one associated with $B(A;3\ell)^c \subset
V$.

For our estimates, it is convenient to calculate the partial trace
as an integral over the group of unitaries \cite{bravyi2006}, see
also \cite{loc_est}. Given an arbitrary observable $A \in \A_V$ and a
set $Y \subset V$, define
\begin{equation}
\langle A \rangle_Y = \int_{\mathcal{U}(Y^c)} U^* A U \, \mu(dU),
\end{equation}
where $\mathcal{U}(Y^c)$ denotes the group of unitary operators over
the Hilbert space $\mathcal{H}_{Y^c}$ and $\mu$ is the associated,
normalized Haar measure. It is easy to see that for any $A \in
\A_V$, the quantity $\langle A \rangle_{Y}$ has been localized to
$Y$ in the sense that $\langle A \rangle_{Y} \in \A_Y$.

Let us define
\begin{equation}
\PB \, = \, \langle P_B( \alpha) \rangle_{B(A;3\ell)}.
\end{equation}
Clearly $\PB$ has the desired support and, moreover, $\|\PB\| \le 1$ since $\|\PB(\alpha)\| \le 1$. 

The difference between $\PB$ and $P_B(\alpha)$ may be written in terms of a
commutator, i.e.  as
\begin{equation} \label{eq:opdif}
\PB \, - \, P_B(\alpha) \, = \, \int_{\mathcal{U}(B(A;3\ell)^c)} U^* \left[
P_B(\alpha),  U  \right] \, \mu(dU).
\end{equation}
Thus, to show that the difference between $\PB$ and $P_B(\alpha)$ is small
(in norm), we need only estimate the commutator of $P_B(\alpha)$ with an
arbitrary unitary supported in $B(A;3\ell)^c$.

This is easy to calculate. Note that
\begin{equation} \label{eq:comm}
\left[ P_B(\alpha), U \right] \, = \, \sqrt{ \frac{ \alpha}{ \pi}} \int_{-
\infty}^{\infty} \left[ e^{i(M_I +M_B+M_E)t} e^{-i(M_I+M_E)t}, U
\right] e^{- \alpha t^2}  \, dt.
\end{equation}
To estimate the integrand, we define the function
\begin{equation}
f(t) \, = \,  \left[ e^{i(M_I +M_B+M_E)t} e^{-i(M_I+M_E)t}, U
\right] .
\end{equation}
A short calculation demonstrates that
\begin{equation} \label{eq:fder}
f'(t) \, = \, i   \left[ e^{i(M_I +M_B+M_E)t} M_B e^{-i(M_I+M_E)t},
U \right] .
\end{equation}
The form of the derivative appearing in (\ref{eq:fder}) suggests
that we define the evolution
\begin{equation}
\alpha_t(A) = e^{i(M_I+M_B+M_E)t}Ae^{-i(M_I+M_B+M_E)t}, \qquad
\mbox{for any local observable } A.
\end{equation}

With this in mind, we rewrite
\begin{eqnarray}
f'(t) & = &  i   \left[ e^{i(M_I +M_B+M_E)t} M_B e^{-i(M_I+M_E)t}, U \right] \nonumber \\
& = &  i   \left[ \alpha_t( M_B) e^{i(M_I+M_B+M_E)t} e^{-i(M_I+M_E)t} , U \right] \nonumber \\
& = &  i \alpha_t( M_B) f(t)  + i  \left[ \alpha_t( M_B), U \right] e^{i(M_I+M_B+M_E)t} e^{-i(M_I+M_E)t} .
\end{eqnarray}

Written as above, the function $f$ can be bounded using
norm-preservation. In particular, let  $V(t)$ be the unitary evolution
that satisfies the time-dependent differential equation
\begin{equation}
i \frac{d}{dt} V(t) =  V(t) \alpha_t(M_B)  \quad \quad \mbox{with }
V(0) = \idty.
\end{equation}
Explicitly, one has that
\begin{equation}
V(t) = e^{i(M_I+M_E)t} e^{-i(M_I+M_B+M_E)t}.
\end{equation}
Considering now the product
\begin{equation}
g(s) = V(s)  f(s)
\end{equation}
it is easy to see that
\begin{equation}
g'(s) = V'(s) f(s) + V(s) f'(s) = i V(s)  \left[ \alpha_s( M_B), U
\right]  V(s)^*.
\end{equation}
Thus
\begin{equation}
 V(t) f(t) = g(t) - g(0) = \int_0^t g'(s) ds,
\end{equation}
and therefore,
\begin{equation}
f(t) \, = \, i  V(t)^* \int_0^t V(s) \left[ \alpha_s( M_B), U \right]
V(s)^* \, ds .
\end{equation}
The bound
\begin{equation} \label{eq:fbd1}
\| f(t) \| \leq \int_0^{|t|} \left\|  \left[ \alpha_s( M_B), U
\right] \right\| \, ds.
\end{equation}
readily follows.

Since the Lieb-Robinson velocity associated to the dynamics $\alpha_t( \cdot)$ grows 
with $\ell$, we estimate (\ref{eq:fbd1}) by comparing back to the original dynamics.
Consider the interpolating dynamics
\begin{equation}
h_s(r) \, = \, \alpha_r \left( \tau_{s-r}(M_B) \right),
\end{equation}
for $0 \leq r \leq s$ and $\tau_r(A) = e^{iH_V r}Ae^{-i H_V r}$. 
With $s$ fixed, it is easy to calculate
\begin{equation}
h_s'(r) \, = \, i \alpha_{s-r} \left( \left[ (M_I+M_B+M_E)-H_V,
\tau_r(M_B) \right] \right).
\end{equation}
We conclude then that
\begin{eqnarray}
\| \alpha_s(M_B) - \tau_s(M_B) \| & = & \left\| \int_0^s h_s'(r) \, dr \right\| \nonumber \\
& \leq & 2 \, \| M_B \| \, \| H_V - (M_I+M_B+M_E) \| \, s \nonumber \\
& \leq & 2 \left( |B(A;2\ell)| J \right) \, \left( 2C_1(\gamma,d,J) | \partial A| \ell^{D} e^{- \ell/ \xi'} \right) \, s \nonumber \\
& \leq & 4^{d+1} J C_1(\gamma,d,J) | \partial A|^2 \ell^{D+d} e^{- \ell/ \xi'} \, s,
\end{eqnarray}
where for the second to last inequality above we used the definition of $M_B$ and Corollary \ref{cor:bd2} and for the last inequality we used Lemma \ref{s_n:bound} to get $|B(A;2\ell)|\le (2\ell) 2^d (2\ell)^{d-1} |\partial A|$.
Hence,
\begin{eqnarray}
\| f(t) \| & \leq & \int_0^{|t|} \left\| \left[ \tau_s(M_B), U \right] \right\| \, ds \, + \,   \int_0^{|t|} \left\| \left[ \alpha_s(M_B) \, - \,  \tau_s(M_B), U \right] \right\| \, ds \nonumber \\
&\leq &  \int_0^{|t|} \left\| \left[ \tau_s(M_B), U \right] \right\| \, ds \nonumber\\
&+& 4^{d+1}J C_1(\gamma,d,J) | \partial A|^2 \ell^{D+d} e^{- \ell/ \xi'} t^2.
\end{eqnarray}
Plugging the above bound into (\ref{eq:comm}), we find that
\begin{eqnarray*}
\| [ P_B(\alpha), U ] \| & \leq & \sqrt{ \frac{ \alpha}{ \pi}} \int_{- \infty}^{\infty} \| f(t) \| e^{- \alpha t^2} dt \nonumber \\
& \leq & \sqrt{ \frac{ \alpha}{ \pi}} \int_{- \infty}^{\infty}\left( \int_0^{|t|} \left\| \left[ \tau_s(M_B), U \right] \right\| \, ds\right) e^{- \alpha t^2} dt \nonumber\\
&+& \frac{J}{\alpha} 2^{2d+1} C_1(\gamma,d,J) | \partial A|^2 \ell^{D+d} e^{- \ell/ \xi'}\label{boundX},
\end{eqnarray*}
since
\be
\sqrt{ \frac{ \alpha}{ \pi}} \int_{- \infty}^{\infty}t^2 e^{- \alpha t^2} dt = \frac{1}{2\alpha}.
\ee
We estimate the double integral above separately for $|t| \ge T$ and $|t| \le T$,
taking $T$ according to (\ref{T}). The integral for $|t| \ge T$ is easily bounded as follows:
\begin{eqnarray}
&&\sqrt{ \frac{ \alpha}{ \pi}} \int_{|t|\ge T} \left( \int_0^{|t|} \left\| \left[ \tau_s(M_B), U \right] \right\| \, ds\right) e^{- \alpha t^2} dt\nonumber\\
&&\le 2\|M_B\| \sqrt{\frac{ \alpha}{ \pi}} \int_{|t|\ge T} |t| e^{- \alpha t^2} dt \nonumber\\
&&\le \frac{2^{2d+1} J}{\sqrt{\alpha \pi}} |\partial A| \ell^{d} e^{-\alpha T^2} \nonumber\\
&&= \frac{2^{2d+1} J}{\sqrt{\pi}\sqrt{\gamma^2+v^2}}  |\partial A| \ell^{d+1/2}e^{-\ell/\xi'}\label{T:large2}
\end{eqnarray}
To bound the integral for $|t| \le T$ we use Lemma \ref{LR:bound}. Remembering that
$\mbox{supp}(U) \subset B(A;3\ell)^c$ and $\mbox{supp}(M_B) \subset B(A;2\ell)$ we have for $|t| \le e^{-(1+\mu)} \ell/v$
\begin{eqnarray}
\|\tau_s(M_B), U\| &\le& 2 \|M_B\| \|U\| |\partial B(A;2\ell)| e^{-\mu \ell} e^{v |s|}\nonumber\\
&\le& 2 \|M_B\| |\partial B(A;2\ell)| e^{-\mu \ell} e^{v |s|} \nonumber\\
&\le& 2 \left(4^d \ell^d |\partial A| J\right) \left(2^{2d-1} \ell^{d-1} |\partial A|\right) e^{-\mu \ell} e^{v |s|}\nonumber\\
&\le& 4^{2d} \,J\, |\partial A|^2\, \ell^{2d-1} e^{-\mu \ell} e^{v |s|}
\end{eqnarray}
Note that as we have already shown in the proof of Proposition \ref{prop:bd2}, since $|t| \le T$, the above bound is useful for $\mu = 1/4$, and $\xi'$ given in the statement of Proposition \ref{prop:bd2}.
Hence,
\begin{eqnarray}
&&\sqrt{ \frac{ \alpha}{ \pi}} \int_{- T}^{T}\left( \int_0^{|t|} \left\| \left[ \tau_s(M_B), U \right] \right\| \, ds\right) e^{- \alpha t^2} dt \nonumber\\
&&\le 4^{2d}\, \frac{J}{v} \, |\partial A|^2\, \ell^{2d-1} e^{-frac{\ell}{4}} \sqrt{ \frac{ \alpha}{ \pi}} \int_{- T}^{T} e^{v|t|} e^{- \alpha t^2} dt\nonumber\\
&&\le 4^{2d}\, \frac{J}{v} \, |\partial A|^2\, \ell^{2d-1} e^{- \frac{\ell}{4}} e^{\frac{v^2}{4\alpha}}\nonumber\\
&&= 4^{2d}\, \frac{J}{v} \, |\partial A|^2\, \ell^{2d-1} e^{-\ell/\xi'}.\label{T:small2}
\end{eqnarray}
Combining (\ref{T:large2}) and (\ref{T:small2}) with (\ref{boundX}) we get:
\begin{eqnarray*}
\| [ P_B(\alpha), U ] \| &\le& \frac{4^dJ}{v} \left(\frac{2v}{\sqrt{\pi}\sqrt{\gamma^2+v^2}} + 4^{d} + \frac{2v}{\gamma^2+v^2} C_1(\gamma,d,J)\right) | \partial A|^2 \,\ell^{D+d+1} e^{- \ell/ \xi'}\\
&\le& C_3(\gamma, d, J) \, | \partial A|^2 \,\ell^{3d} e^{- \ell/ \xi'},
\end{eqnarray*}
where we set 
$$C_3(\gamma, d, J) = \frac{4^dJ}{v} \left(\frac{2v}{\sqrt{\pi}\sqrt{\gamma^2+v^2}} + 4^{d} + \frac{2v}{\gamma^2+v^2} C_1(\gamma,d,J)\right)$$ and used the bound $D \le 2d-1$ in the last line.
The estimate claimed in (\ref{eq:locpbd}) now easily follows from (\ref{eq:opdif}).
\end{proof}
        

    %
    %

    \newchapter{Multiplicativity of Quantum Channels}{Multiplicativity of Depolarizing Werner-Holevo Channels}{Multiplicativity of the maximal $2$-norm for depolarizing Werner-Holevo Channels}
    \label{sec:LabelForChapter4}

        
    
\section{Introduction}
In this chapter\footnote{Reprinted with permission from Journal of Mathematical Physics, {\bf 48}, 122102. Copyright 2007, American Institute of Physics.}, we explore the question of multiplicativity for the output $2$-norm of a class of
quantum channels related closely to the $d$-dimensional Werner-Holevo (WH) channel ${\cal W}_d(\rho) = \frac{1}{d-1} ((\Tr(\rho)) \one_d - \rho^T)$, which is known~\cite{WH} to give a counterexample to the multiplicativity of the maximal output $p$-norm for $p > 4.79$, when $d=3$. For large dimensions, the WH channel acts like the completely depolarizing channel, since 
$\|{\cal W}_d(\rho)- \frac{1}{d} \one_d\|_{\infty} = 1/d$. This is a property shared by the recent counterexamples to multiplicativity for $p > 2$ by Winter and for $1 < p < 2$, by Hayden~\cite{Winter,Hayden}.
Nevertheless, it has been shown~\cite{Datta, AF} that ${\cal W}_d(\rho)$ satisfies multiplicativity for $1\leq p \leq 2$.

\section{The setup and some useful lemmas} 
It is a natural extension of the previous research on constructing counterexamples to multiplicativity, 
to study the output $p$-norm of channels of the form
$$\W(\rho) = \lambda \rho + (1-\lambda) {\cal W}_d(\rho),$$ and ask if those channels
satisfy multiplicativity for $p$-norms with $p = 2$. We adopt the name ``depolarized" Werner-Holevo channels due to the close connection of $\W$ with depolarizing channels in higher dimensions.
We focus our attention to the study of the output $2$-norm for the tensor product channel $\W\otimes\W$
acting on bipartite states in $M_d(\Cx)\otimes M_d(\Cx)$ and show that multiplicativity is satisfied for this norm for all dimensions $d$.

A direct computation of the eigenvalues of $\W\otimes\W(\pure{\psi_{12}})$ turns out to be much harder for $0 < \lambda < 1$, than for the boundary cases $\lambda = 0, 1$. The reason is that the
output consists of a combination of the input state and its transpose/partial transpose, which in
general do not share a common eigenbasis. To work around this difficulty, we compute explicitly the
output $2$-norm of $\W\otimes\W(\pure{\psi_{12}})$ and the maximal output $2$-norm of $\W$ and study the difference 
\begin{equation}\label{eqn:output_dist}
\D(\psi_{12}) = (\| \W \|_2^2)^2 - \|\W \otimes \W(\pure{\psi_{12}}) \|_2^2
\end{equation}
We show that $\D \ge 0$ for all input states and $\lambda \in [0,1],\,d\ge 2$.
We begin with the computation of $\| \W \|_2^2$ in the following Lemma.

\begin{lemma}\label{lem:max_norm}
The (squared) maximal output $2$-norm of $\W$ is given by
$$\|\W\|_2^2 = \frac{(d-2)\lambda^2 + 1}{d-1}$$
\end{lemma}
\begin{proof}
It is easy to check that $\|\W(\pure{\psi})\|_2^2$ is
\begin{eqnarray*}
&=& \Tr(\W(\pure{\psi})^2) \\
&=& \lambda^2 + \frac{2\lambda(1-\lambda)(1-|\braket{\psi}{\overline{\psi}}|^2)}{d-1} + \frac{(1-\lambda)^2}{d-1}\\
&\le& \frac{(d-2)\lambda^2 +1}{d-1},
\end{eqnarray*}
where $\ket{\overline{\psi}}$ denotes the complex conjugate of $\ket{\psi}$ in the standard basis.
Taking $\ket{\psi} = \frac{\ket{0} + i\ket{1}}{\sqrt{2}}$, with $\ket{0},\ket{1}$ two standard basis vectors, we see that equality can be achieved in the above expression and the result follows.
\end{proof}

We now turn our attention to the more complicated output $2$-norm of $\W\otimes\W(\pure{\psi_{12}})$.

\begin{lemma}\label{lem:output_norm}
The (squared) output $2$-norm $\|\W\otimes\W(\pure{\psi_{12}})\|_2^2$ is given by:
\begin{eqnarray*}
& &(\|\W\|_2^2)^2\\ 
&+& S_\lambda^2 \, |\braket{\psi_{12}}{\overline{\psi_{12}}}|^2\\
&-& 2 (S_\lambda + R_\lambda^2)
(S_\lambda + (d-2)Q_\lambda^2\Big)
\big(1-\|\rho_1\|_2^2\big)\\
&-& S_\lambda \|\W\|_2^2 \Tr(\rho_1\rho_1^T + \rho_2\rho_2^T),
\end{eqnarray*}
where $Q_\lambda = \frac{1-\lambda}{d-1}$, $R_\lambda = \lambda - Q_\lambda$, $S_\lambda = 2\lambda Q_\lambda$, $\rho_1 = {\Tr}_2\pure{\psi_{12}}$, $\rho_2 = {\Tr}_1\pure{\psi_{12}}$ and $T$ denotes
transposition.
\end{lemma}
\begin{proof}
It is easy to check that 
\begin{eqnarray*}
\W\otimes\W(\pure{\psi_{12}}) &=& \lambda^2 \pure{\psi_{12}}\\
&+& Q_\lambda R_\lambda [\rho_1\otimes\one_d + \one_d\otimes\rho_2]\\
&+& Q_\lambda^2\big[\one_d\otimes\one_d + \pure{\overline{\psi_{12}}}\big]\\
&-& \frac{S_\lambda}{2}(\pure{\psi_{12}})^{T_1}\\
&-& \frac{S_\lambda}{2}(\pure{\psi_{12}})^{T_2},
\end{eqnarray*}
where $T_1, T_2$ denote partial transposition w.r.t. the $1^{\rm st}$, $2^{\rm nd}$ tensor factor, respectively.
Taking the trace after squaring the above expression and noting that 
$$\Tr \pure{\psi_{12}}(\pure{\psi_{12}})^{T_k} = \Tr \rho_k \rho_k^T,$$ for $k=1,2$ (which one can show using the Schmidt decomposition of $\ket{\psi_{12}}$), we get the desired result.
\end{proof}

The following observations will be very useful in the proof of the main theorem, so we state them
here as lemmas.
\begin{lemma}\label{ineq:two_norm}
Let $\sigma_1 \le \sigma_2 \le \ldots \le \sigma_d$ be non-negative numbers that sum up to $1$. Then, the following inequality holds:
\begin{equation*} \sigma_d \ge \sum_{i = 1}^d \sigma_i^2
\end{equation*}
\end{lemma}
\begin{proof}
The r.h.s. of the inequality can be thought of as the expected value of the random variable $X$ given by $\mathrm{Pr}(X = \sigma_i) = \sigma_i$. The upper bound then follows immediately.
\end{proof}

\begin{lemma}\label{ineq:conjugate}
Let $\Phi$ be a quantum channel and denote by $\overline{\Phi}$ the complex conjugate (w.r.t. the Kraus operators) channel. 
Then, the following inequality holds:
\begin{equation*} \|\Phi\otimes\Phi\|_2 \le \|\Phi\otimes\overline{\Phi}\|_2
\end{equation*}
More importantly, the maximal output $2$-norm on the r.h.s. is achieved on inputs with conjugate Schmidt bases: 
\begin{equation*}\ket{\psi_{max}} = \sum_i \sqrt{\sigma_i} \ket{e_i}\overline{\ket{e_i}},
\end{equation*}
where the conjugate is taken w.r.t. the standard basis.
\end{lemma}
\begin{proof}
Let $\ket{\psi_{12}} = \sum_i \sqrt{\sigma_i} \ket{e_i}\ket{f_i}$ be the Schmidt decomposition of the input. Computing $\|\Phi\otimes\Phi(\pure{\psi_{12}})\|_2^2$ and setting $\alpha_{i,j,k,l} = \Tr(\Phi(\ketbra{e_i}{e_j})\Phi(\ketbra{e_k}{e_l}))$,
$\beta_{i,j,k,l} = \Tr(\Phi(\ketbra{f_i}{f_j})\Phi(\ketbra{f_k}{f_l}))$ and $\sigma_{i,j,k,l} = \sigma_i\sigma_j\sigma_k\sigma_l$, we get:
\begin{eqnarray*}
\|\Phi\otimes\Phi(\pure{\psi_{12}})\|_2^2 &=& \sum_{i,j,k,l} \sqrt{\sigma_{i,j,k,l}}\, \alpha_{i,j,k,l}\,\beta_{i,j,k,l}\\
&=& \sum_{i,j,k,l} \sqrt{\sigma_{i,j,k,l}}\, {\mathcal Re}\{\alpha_{i,j,k,l}\,\beta_{i,j,k,l}\}\\
&\le& \sum_{i,j,k,l} \sqrt{\sigma_{i,j,k,l}}\, \frac{|\alpha_{i,j,k,l}|^2}{2}\\
&+& \sum_{i,j,k,l} \sqrt{\sigma_{i,j,k,l}}\,\frac{|\beta_{i,j,k,l}|^2}{2}\\
&\le& \sum_{i,j,k,l} \sqrt{\sigma_{i,j,k,l}}\, |\alpha_{i,j,k,l}|^2,
\end{eqnarray*}
where we assumed the last inequality w.l.o.g.
But, the last expression is equal to $\|\Phi\otimes\overline{\Phi}(\pure{\psi_{max}})\|_2^2$. Moreover, one can follow the above
steps to show that for every input $\ket{\psi_{12}}$ there is an input $\ket{\psi_{max}'}$ such that
\begin{equation*}
\|\Phi\otimes\overline{\Phi}(\pure{\psi_{12}})\|_2^2 \le \|\Phi\otimes\overline{\Phi}(\pure{\psi_{max}'})\|_2^2.
\end{equation*}
This concludes the proof.
\end{proof}
At this point, it is important to note that the above lemma and the fact that $\|\Phi\|_2 = \|\overline{\Phi}\|_2$ imply that 
$\|\Phi\otimes\Phi\|_2 \le \|\Phi\|_2^2$ follows from $\|\Phi\otimes\overline{\Phi}(\pure{\psi_{max}})\|_2 \le \|\Phi\|_2^2$.
\section{Proof of the Main Result}
In this section, we will show that the difference $\D$ defined in~(\ref{eqn:output_dist}) is always non-negative, which is equivalent to multiplicativity of the output $2$-norm for $\W$. We state this as a theorem:
\begin{theorem}
For the depolarized Werner-Holevo channel $\W$, we have for $\lambda \in [0,1], d\ge 2$:
$$\|\W\otimes\W\|_2 = \|\W\|_2^2$$
\end{theorem}

\begin{proof}
From Lemma~\ref{lem:output_norm} we see that the condition $\D(\ket{\psi_{12}}) \ge 0$ is equivalent to
\begin{eqnarray*}
S_\lambda^2\, 
|\braket{\psi_{12}}{\overline{\psi_{12}}}|^2 &\le&
2 (S_\lambda^2 + P_\lambda^2)
\big(1-\|\rho_1\|_2^2\big)\\
&+& S_\lambda \|\W\|_2^2 \Tr(\rho_1\rho_1^T + \rho_2\rho_2^T),
\end{eqnarray*}
where $P_\lambda^2 = [Q_\lambda^2+(d-2)R_\lambda^2]S_\lambda
+ (d-2)Q_\lambda^2 R_\lambda^2 \ge 0$.
Using Lemma~\ref{lem:max_norm} to write $\|\W\|_2^2$ as $(1+\sqrt{d-1})S_\lambda + (\lambda - \frac{1-\lambda}{\sqrt{d-1}})^2$, we see that it is sufficient to prove the following inequality
\begin{equation}\label{ineq:output_dist}
|\braket{\psi_{12}}{\overline{\psi_{12}}}|^2 \le
2\big(1-\|\rho_1\|_2^2\big) + (1+\sqrt{d-1})\Tr(\rho_1\rho_1^T + \rho_2\rho_2^T)
\end{equation}
(the boundary cases $\lambda = 0,1$ follow from $1 \ge \|\rho_1\|_2^2$).

Since both the identity and the WH channel remain unchanged under complex conjugation, their convex combination inherits that property.
With this in mind, Lemma~\ref{ineq:conjugate} implies that we only need to check if:
\begin{equation}\label{ineq:output_dist2}
|\braket{\psi_{max}}{\overline{\psi_{max}}}|^2 \le
2\big(1-\|\rho_1\|_2^2\big) + 2(1+\sqrt{d-1})\Tr(\rho_1\rho_1^T),
\end{equation}
since $\rho_1 = \rho_2^T$ for $\ket{\psi_{max}} = \sum_i \sqrt{\sigma_i} \ket{e_i}\overline{\ket{e_i}}$.
Moreover, $\|\rho_1\|_2^2 = \sum_{i = 1}^d \sigma_i^2$, where some of the $\sigma_i$ may be zero. Applying Lemma~\ref{ineq:two_norm} (and borrowing its notation w.l.o.g.), it follows that $\|\rho_1\|_2^2 \le \sigma_d$. 
It becomes clear now that in order to prove~(\ref{ineq:output_dist}), it is sufficient to show:
\begin{equation}\label{ineq:output_dist3}
|\braket{\psi_{max}}{\overline{\psi_{max}}}|^2 \le
2\big(1-\sigma_d) + 2(1+\sqrt{d-1})\Tr(\rho_1\rho_1^T)
\end{equation}
for $\sigma_d \ge 1/2$, since $|\braket{\psi_{max}}{\overline{\psi_{max}}}| \le 1$ and 
$\Tr(\rho_1\rho_1^T) \ge 0$.
We now compute the following two quantities:
\begin{eqnarray}
|\braket{\psi_{max}}{\overline{\psi_{max}}}| &=& \sum_{i,j} \sqrt{\sigma_i \sigma_j}
|\braket{e_i}{\overline{e_j}}|^2\label{eqn:inner_transpose}\\
\Tr(\rho_1 \rho_1^T) &=& \sum_{i,j} \sqrt{\sigma_i \sigma_j}
|\braket{e_i}{\overline{e_j}}|^2\label{eqn:outer_transpose} 
\end{eqnarray}
We will need to treat dimensions $d \le 4$ and $d \ge 5$ separately.
For $d \le 4$ we use Cauchy-Schwarz to get the following estimate for
$$\Big(\sum_{i,j} \sqrt{\sigma_i \sigma_j}
|\braket{e_i}{\overline{e_j}}|^2\Big)^2
$$
\begin{eqnarray}
&\le& \Big(\sum_{i,j} \sigma_i \sigma_j |\braket{e_i}{\overline{e_j}}|^2\Big) 
\Big(\sum_{i,j} |\braket{e_i}{\overline{e_j}}|^2\Big) \nonumber \\
&\le& d \sum_{i,j} \sigma_i \sigma_j
|\braket{e_i}{\overline{e_j}}|^2 \nonumber\\
&=& d \Tr(\rho_1 \rho_1^T) \label{cs1}
\end{eqnarray}
where we have used Parseval's identity in the last inequality. Using~(\ref{ineq:output_dist3}) and~(\ref{eqn:inner_transpose}) we see from estimate~(\ref{cs1}) that it is sufficient to show that $d \leq 2(1+\sqrt{d-1})$, which is true for $d \le 4$.
We now turn our attention to the case $d \ge 5$.
We will need a different estimate than the one given in~(\ref{cs1}), since we need to make use
of the assumption that $\sigma_d \ge 1/2$ in order to lower the factor $d$ in~(\ref{cs1}).
We start by using Cauchy-Schwarz to get the following upper bound:
\begin{equation}\label{cs2}
\Big(\sum_{i,j} \sqrt{\sigma_i \sigma_j}
|\braket{e_i}{\overline{e_j}}|^2\Big)^2
\le 3 (I_1^2 + I_2^2 + I_3^2),
\end{equation}
where 
\begin{eqnarray*}
I_1 &=& \sum_j \sqrt{\sigma_d \sigma_j}
|\braket{e_d}{\overline{e_j}}|^2\\
I_2 &=& \sum_i \sqrt{\sigma_i \sigma_d}
|\braket{e_i}{\overline{e_d}}|^2\\
I_3 &=& \sum_{i \ne d, j \ne d} \sqrt{\sigma_i \sigma_j}
|\braket{e_i}{\overline{e_j}}|^2
\end{eqnarray*}
A further application of Cauchy-Schwarz on $I_1 = I_2$ and $I_3$ will give us the desired result.
We start with an estimate for $I_1$.
Noting that one of the summation indices is fixed to $d$, we get
\begin{eqnarray*}
I_1^2 &=& \Big(\sum_{j} \sqrt{\sigma_d \sigma_j}
|\braket{e_d}{\overline{e_j}}|^2\Big)^2 \\
&\le& \Big(\sum_{j} \sigma_d \sigma_j
|\braket{e_d}{\overline{e_j}}|^2 \Big)\Big(\sum_{j} |\braket{e_d}{\overline{e_j}}|^2\Big)\\
&=& \sum_{j} \sigma_d \sigma_j |\braket{e_d}{\overline{e_j}}|^2\\
&\le& \Tr(\rho_1 \rho_1^T)
\end{eqnarray*}
Similarly, we have $I_2^2 \le \Tr(\rho_1 \rho_1^T)$. Since $1+\sqrt{d-1} \ge 3$ for $d \ge 5$, we
see from~(\ref{ineq:output_dist3}) and~(\ref{cs2}) that it remains to show $3\, I_3^2 \le 2\,(1-\sigma_d)$.
We have
\begin{eqnarray*}
I_3^2 &=& \Big(\sum_{i \ne d,j \ne d} \sqrt{\sigma_i \sigma_j}
|\braket{e_i}{\overline{e_j}}|^2\Big)^2\\
&\le& \Big(\sum_{i \ne d,j \ne d} \sigma_i
|\braket{e_i}{\overline{e_j}}|^2 \Big)^2\\
&\le& \Big(\sum_{i \ne d} \sigma_i \Big)^2\\
&=& (1-\sigma_d)^2
\end{eqnarray*}
It remains to show that $3\, (1-\sigma_d)^2 \le 2 \, (1-\sigma_d) 
\Leftrightarrow (1-\sigma_d)(3\sigma_d -1) \ge 0$, which follows from our earlier observation that we only need to consider $\sigma_d\in[\frac{1}{2},1]$.
\end{proof}

\section{Beyond entrywise positivity}
In this section, we show that the depolarized Werner-Holevo channels satisfy multiplicativity despite
having a more complex structure than quantum channels satisfying a condition sufficient for multiplicativity of the maximal output $2$-norm.
\begin{proposition}
The depolarized Werner-Holevo channels $\W$ with $\lambda \in (0,1)$ and  $d \ge 3$ do not satisfy the following entrywise-positivity (EP) condition: There exists an orthonormal basis $\{\ket{e_i}\}_{i=1}^d$ of $\Cx^d$ such that
\begin{equation*}
\Tr \W(\ketbra{e_l}{e_i})\W(\ketbra{e_j}{e_k}) \ge 0, \qquad \forall i,j,k,l.
\end{equation*}
\end{proposition}
\begin{proof}
One can check that $\Tr \W(\ketbra{e_l}{e_i})\W(\ketbra{e_j}{e_k})$ is given by:
\begin{equation*}
\Big[\lambda^2+\Big(\frac{1-\lambda}{d-1}\Big)^2\Big] \delta_{i,j}\delta_{k,l}
+ \Big[\frac{2\lambda(1-\lambda)}{d-1}+(d-2)\Big(\frac{1-\lambda}{d-1}\Big)^2\Big] \delta_{i,l}\delta_{j,k}
- \frac{2\lambda(1-\lambda)}{d-1} \braket{e_i}{\overline{e_k}}\braket{\overline{e_l}}{e_j}
\end{equation*}
where $\ket{\overline{e_k}}$ denotes the complex conjugate of $\ket{e_k}$, as before.
Now, taking $i=j, k\ne l$ in the above expression, we see that the EP condition implies:
$$ \braket{e_i}{\overline{e_k}}\braket{\overline{e_l}}{e_i} \le 0, \qquad \forall i, k\ne l.$$
Summing over $i$ in the above inequality gives us $0$, which implies that:
\be\label{cond}
\braket{e_i}{\overline{e_k}}\braket{\overline{e_l}}{e_i} = 0, \qquad \forall i, k\ne l.
\ee
Fixing $l$, we choose $i = \pi(l)$ such that $\braket{\overline{e_l}}{e_{\pi(l)}} \ne 0$ (we can always find such a $\pi(l)$, since otherwise $\ket{\overline{e_l}} = 0$; a contradiction to $\ket{\overline{e_l}}$ being an orthonormal basis vector). 
Condition~(\ref{cond}) then implies that $\braket{e_{\pi(l)}}{\overline{e_k}} = 0, \, \forall k\ne l$. 
Since the $\{\ket{\overline{e_k}}\}$ form an orthonormal basis, it follows that 
$\ket{e_{\pi(l)}} = \ket{\overline{e_l}}, \forall l$.
We may now rewrite the EP condition as:
\begin{equation*}
\Big[\lambda^2+\Big(\frac{1-\lambda}{d-1}\Big)^2\Big] \delta_{i,j}\delta_{k,l}
+ \Big[\frac{2\lambda(1-\lambda)}{d-1}+(d-2)\Big(\frac{1-\lambda}{d-1}\Big)^2\Big] \delta_{i,l}\delta_{j,k}
\ge \frac{2\lambda(1-\lambda)}{d-1} \delta_{i, \pi(k)}\delta_{j, \pi(l)}
\end{equation*}
Choosing $i = \pi(k), j = \pi(l)$ and $k \ne l$, the above condition becomes:
\begin{equation*}
\Big[\frac{2\lambda(1-\lambda)}{d-1}+(d-2)\Big(\frac{1-\lambda}{d-1}\Big)^2\Big] \delta_{\pi(k),l}\delta_{\pi(l),k}
\ge \frac{2\lambda(1-\lambda)}{d-1}
\end{equation*}
The EP condition forces $\pi(l) = k, \forall k \ne l$ (note that $\pi(k) = l$ then follows from the definition of $\pi(k)$,) which is impossible for $d \ge 3$. For $d=2$, choosing $\ket{e_1} = \frac{\ket{0} + i \ket{1}}{\sqrt{2}}$ satisfies the EP condition.
\end{proof}

\section{Discussion}
We have shown that for depolarized Werner-Holevo channels the maximum output $2$-norm is
multiplicative. For $\lambda \in (0,1)$ and $d \ge 3$, the depolarized Werner-Holevo maps do not satisfy the entrywise-positivity (EP) condition introduced by C. King and M.B. Ruskai in~\cite{KR1, KR2}, which makes our result non-trivial. Moreover, a closer look at Lemma~\ref{ineq:conjugate} in conjuction with random unitary channels may yield a counterexample for the multiplicativity of the maximal output 2-norm in the vein of \cite{Winter}.
    

    %
    %

\end{document}